\documentclass[10pt,twocolumn,twoside]{IEEEtran}
\usepackage{cite,url}
\usepackage{graphicx,subfigure,float,color,bigints,mathrsfs}
\usepackage[ruled, vlined]{algorithm2e}
\usepackage{comment}
\usepackage{standalone,enumerate}
\usepackage{amsmath,amssymb,textcomp}
\usepackage{tikz}
\usetikzlibrary{graphs,arrows,decorations.pathreplacing,intersections,positioning,calc}
\usetikzlibrary{datavisualization.formats.functions,datavisualization}
\usepackage{pgfplots}	
\usepackage{tikzscale}
\usepackage[noend]{algpseudocode}
\usepackage{dsfont,bbm}
\usepackage[active]{srcltx}
\usepackage[T1]{fontenc}
\usepackage{authblk}
\usepackage{kantlipsum}
\allowdisplaybreaks
\DeclareMathOperator*{\argmax}{arg\,max}

\newcommand{\BlackBox}{\rule{1.5ex}{1.5ex}}  
\newenvironment{proof}{\par\noindent{\bf Proof\
}}{\hfill\BlackBox\\[2mm]}
\newtheorem{theorem}{\bf{Theorem}}
\newtheorem{lemma}{\bf{Lemma}}

\newtheorem{remark}{\bf{Remark}}

\newcommand{\be}{\begin{equation}}
\newcommand{\ee}{\end{equation}}
\newcommand{\bea}{\begin{eqnarray*}}
\newcommand{\eea}{\end{eqnarray*}}

\makeatletter
\def\Lddots{\mathinner{\mkern1mu\raise17\p@\vbox{\kern17\p@\hbox{.}}\mkern2mu
    \raise8\p@\hbox{.}\mkern2mu\raise\p@\hbox{.}\mkern1mu}}
 \makeatother

\outer\def\subsect#1\par{\vskip12pt
plus.07\vsize\penalty-250\vskip0pt plus-.07\vsize
\bigskip\vskip\parskip\message{#1}
\vbox{\smash{\lower9pt\hbox{\kern-8pt\epsfbox{shadedbox.eps}}}}\vskip-\baselineskip
\leftline{\large\bf#1}\nobreak\medskip}

\def\BibTeX{{\rmfamily B\kern-.05em{\scshape i\kern-.025em b}\kern-.08em \TeX}}

\newcommand{\A}{\mathcal{A}}

\newcommand{\C}{\mathscr{C}}
\newcommand{\D}{\mathcal{D}}

\newcommand{\N}{\mathcal{N}}

\newcommand{\bS}{\mathbf{S}}

\newcommand{\bs}{\mathbf{s}}

\newcommand{\btau}{\mbox{\boldmath $\tau$}}
\newcommand{\bphi}{\mbox{\boldmath $\phi$}}
\newcommand{\bomega}{\mbox{\boldmath $\omega$}}
\newcommand{\bdelta}{\mbox{\boldmath $\delta$}}

\newcommand{\bPi}{\mbox{\boldmath $\Pi$}}

\makeatletter
\newcommand*{\coloneq}{\mathrel{\rlap{%
                     \raisebox{0.3ex}{$\m@th\cdot$}}%
                     \raisebox{-0.3ex}{$\m@th\cdot$}}%
                     =}
\makeatother
\begin{document}
\newcounter{mytempeqncnt}
\title{Utility Maximizing Sequential Sensing\\ Over a Finite Horizon}
\author{Lorenzo Ferrari, Qing Zhao, and Anna Scaglione
\IEEEcompsocitemizethanks{\IEEEcompsocthanksitem L. Ferrari and A. Scaglione \{Lorenzo.Ferrari,Anna.Scaglione@asu.edu\} are with the School of Electrical, Computer and Energy Engineering, Arizona State University, Tempe,
AZ, 85281, USA.\protect
\IEEEcompsocthanksitem Q.Zhao \{qz16@cornell.edu\} is with the School of Electrical Engineering, Cornell University,
Ithaca, NY, 14853, USA.\protect
}
\thanks{This work was supported in part by the National Science Foundation CCF-247896}}
\maketitle
\IEEEpeerreviewmaketitle
\begin{abstract}
We consider the problem of optimally utilizing $N$ resources, each in an unknown binary state. The state of each resource can be inferred from state-dependent noisy measurements. Depending on its state, utilizing a resource results in either a reward or a penalty per unit time. The objective is a sequential strategy governing the decision of sensing and exploitation at each time to maximize the expected utility (i.e., total reward minus total penalty and sensing cost) over a finite horizon $L$. We formulate the problem as a Partially Observable Markov Decision Process (POMDP) and show that the optimal strategy is based on two time-varying thresholds for each resource and an optimal selection rule for which resource to sense. Since a full characterization of the optimal strategy is generally intractable, we develop a low-complexity policy that is shown by simulations to offer near optimal performance. This problem finds applications in opportunistic spectrum access, marketing strategies and other sequential resource allocation problems.
\end{abstract}


\section{Introduction}
\label{sec:introduction}
We consider the problem of optimally utilizing $N$ resources over a horizon of length $L$. The state of each resource is either $0$ (``good'') or $1$ (``bad'') and is unknown \emph{a priori}. Utilizing resource $i$ results in either a reward $r_i$ or a penalty $\rho_i$ per unit time, depending on its state. To infer the states of the resources, only one of them can be sensed at each given time, accruing a random measurement drawn from a distribution determined by the state of the chosen resource. We study the optimal sequential strategy governing the decision of sensing and exploitation at each time, that maximizes the expected utility (i.e., total reward minus total penalty) over the finite horizon. Since sensing reduces the available time for utilization, the essence of the problem is in balancing the overhead associated with sensing and the penalty of utilizing resources in a bad state. Due to the limited sensing capability (i.e., only one resource can be sensed at a given time), it is also crucial to choose judiciously which resource to sense at each time. This problem arises in cognitive radio systems where the resources are channels that can be either busy or idle. 
\vspace{-.3cm}
\subsection{Main Results}
We formulate the problem above as a partially observable Markov decision process (POMDP) with state given by the current time, the set of remaining resources for which a utilization decision has not been made, and the belief on the state of these resources (i.e., the posterior probability that a resource is in state $0$) in Section \ref{sec:formulation}. A policy of this POMDP consists of three components: a set of selection rules indicating which resources to sense at each time, a set of stopping rules, and the induced stopping times, indicating when to terminate the sensing phase of each resource, and a set of decision rules governing whether to exploit or discard a resource after sensing (see Section \ref{subsec:POMDP}). 
\par In Section \ref{sec:design} we explicitly characterize the optimal decision rules and 
 show that the optimal stopping rule is given by two time-varying thresholds that also depend on the resources with pending decision and,
since characterizing the time-varying thresholds and the optimal resource selection rules is analytically intractable, we develop a low-complexity suboptimal policy
(described in \ref{subsec:heuristic}).
Based on insights gained from the monotone properties of the thresholds in the single-resource case, we develop a recursive algorithm for computing approximate thresholds. For the resource selection rule, we propose an index policy based on an interchange argument. This index rule prioritizes resources according to the integrated effect of their utilization reward $r_i$, belief values, and the quality of the sensing measurements (which determines the speed of the state inference). 
In Section \ref{sec:applications}, we showcase two applications of our framework: cognitive radio and marketing strategy. In Section \ref{sec:simulation} we present simulation results to corroborate our theoretical findings.
The asymptotic optimality of the index selection rule is also discussed.
\vspace{-.3cm}
\subsection{Related Work}
A majority of existing work on sequential sensing focuses on typical detection metrics: minimizing the expected detection delay subject to constraints on detection accuracy in terms of the probabilities of false alarm and miss detection. This body of work can be partitioned into two general categories: passive sequential hypothesis testing and active sequential hypothesis testing. The former was pioneered by
Wald in 1947\cite{Wald:1947}, { which introduced the procedure known as Sequential Probability Ratio Test (SPRT)}. Under this formulation, the observation model under each hypothesis is predetermined, and a test only needs to determine when to stop taking observations and which hypothesis to declare at the time of stopping. The latter was pioneered by Chernoff in 1959\cite{chernoff1959}.  Active hypothesis testing has a control aspect that allows the decision maker to choose the experiment to be conducted at each
time. Different experiments generate observations from different distributions under each hypothesis. An active test thus includes a selection rule (i.e., which experiment to carry out at each given time) in addition to a stopping rule and a declaration rule. Following these seminal results, there has been an extensive body of work on both passive (see, for example, \cite{arrow1949bayes,woodroofe,sequential_analysis_book,
decentralized_sequential_test,multiscale_bayesian_SPRT,sahu_distributed} and references therein) and active (see \cite{bessler1960theory,paper_convexity_proof,nitina_controlled_multi,naghshvar2013,naghshvar_sequential,kobe_active} and references therein) hypothesis testing. 
\par There are a couple of studies on sequential sensing for anomaly detection under the objective of minimizing operational cost \cite{paper_kobe_indexing,Cohen&Zhao:15TSP} which can be considered as a utility function. Different from this paper, these studies either restrict admissible sensing strategies to those that declare the state of one process before starting sensing another process (i.e., no switching across processes) \cite{paper_kobe_indexing} or focus on asymptotically (when the horizon length approaches infinity) optimal policies. 

In the Cognitive Radio literature our utility maximization problem can be seen as a parametric formulation of the energy efficiency metric maximization \cite{eneeff_Wu_Xu_Chen,eneeff_Pei_Liang_Li}.
In \cite{eneeff_Wu_Xu_Chen} the presence of PU communication is inferred via energy detection (known to be optimal when no prior information is available) and a single PU channel was considered. 
Multiple PU channels and the capability of switching between different channels, known as {\it spectrum handoff}, is investigated in \cite{spectrum_handoff_wang_vehic} and an efficient convex optimization procedure is developed to solve for the optimal values of the sensing slot duration and the channel switching probability in order to minimize energy consumption while guaranteeing satisfying throughput.
In \cite{eneeff_Pei_Liang_Li} prior knowledge (the vector $\boldsymbol{\omega}$ in our model) over the state of different channels is considered, but the sequential decision process terminates when the SU decides to transmit over one single channel. The problem is formulated as DP but no heuristic is provided to tackle the combinatorial complexity and no further insight on the threshold structure of the decision is proposed. 
The threshold structure for the channel probing stopping rule that has been proved in \cite{optimal_channel_probing} and that \cite{eneeff_Pei_Liang_Li} refers to, also considers only one possible transmission and a constant data time scenario, i.e. the transmission time is not affected by the time spent in sensing the channels.
More general utility maximization approaches for cognitive networks can be found in
\cite{cognitive_um_zheng} and \cite{utility_CRSN_zheng}, which leverage the class of utility maximization functions for optimal congestion and contention control introduced in \cite{NUM_cc_Calderbank}.
 A censored truncated sequential spectrum sensing procedure is presented in \cite{censored_truncated_maleki}\footnote{A review on censoring and sequential procedures can be found in the references of \cite{censored_truncated_maleki}}, where different cognitive radios sense the same channel and decide whether to send their estimates to a fusion center that then performs the final decision over the presence of a PU.
We instead consider a single cognitive radio that can sense different channels (but only one at a time), which could be seen as using only one sensor per channel, and therefore each sensor is sensing a different state. 
The limit of one sensing operation at a time, in our formulation, could be seen as a rigid censoring constraint, with the possibility of suspending the decision over a channel, while continuing to sense and/or exploit others, and then potentially reconsider whether to sense or exploit that channel after some time instants.
A concatenation of SPRT is proposed as a low-complexity, asymptotically optimal policy in  \cite{paper_concatenated_SPRT}.
\par Additional relevant works can be found in \cite{Sensing-Throughput,POMDP_Qing,survey,diversity_ICASSP,optimal_linear_cooperation,survey2,coop_comm_for_CRN,eff_discovery,eigenvalue,saeed_restless}. The vast majority considers a rigid separation between exploration and exploitation phase, while in our framework we enable a combination of the two over different resources, by accessing some channels and simultaneously sensing a different resource at each time.
For the application projection of our work, we will consider spectrum sensing for cognitive radios and market strategy with online advertisement.
Notice that in both these applications, several works have adopted a Multi-Armed Bandit (MAB) formulation (see for example \cite{saeed_restless,gai2010learning} for cognitive radios , \cite{NIPS2008_3580,lu2010contextual} for online advertising), whereas others (including this work) followed a Partially Observable Markov Decision Process (POMDP) framework \cite{POMDP_Qing,DBLP:journals/corr/KrishnamurthyAB16}. 
It is important to highlight the following difference: 
in the MAB formulation the utility obtained by the player that selects a certain "arm" (i.e. an action) is the only information that is used to optimize the future choices. 
The concept of POMDP, instead, can be used to cast a wider class of decision problems where at each time epoch a certain action needs to be designed, that provides indirect information on what strategy the player should use to harness the maximum utility. In other words a 'sensing' action informs what the player should do. The MAB formulation is a special case of the POMDP in which the sensing action and the action that brings the utility to the player are the same thing. 
Therefore, the POMDP formulation allows in principle for a richer action space than a MAB problem and a POMDP cannot necessarily be mapped into a MAB problem.
\par
The key difference between this work and the vast body of results on hypothesis testing is that the design objective in the problem studied here is the utility maximization that directly addresses the trade-off between exploration (i.e., detecting the state of each resource) and the time that remains for exploitation (of the resources, based on the information gathered during the sensing phase). Hypothesis testing problems, passive or active, are pure exploration problems.
With respect to the works in \cite{paper_kobe_indexing}-\cite{Cohen&Zhao:15TSP}, the problem considered in this paper allows switching across resources and focuses on a finite horizon, resulting in a full-blown POMDP that is significantly more difficult to tackle (for a discussion on the general complexity of POMDP see \cite{pomdp_complexity_mundhenk}).
In \cite{eneeff_Wu_Xu_Chen}-\cite{spectrum_handoff_wang_vehic} no prior information is available and all the channels are equal. Therefore, there is no ordering of the channels to take into account and no SPRT procedure to be optimized for the sensing (performed via energy detection with a deterministic sensing time), whereas our paper tackles both aspects.
The difference of our model with respect to a concatenation of truncated SPRT (and therefore a standard truncated SPRT) \cite{paper_concatenated_SPRT} is highlighted at the end of Section \ref{sec:formulation}, considering the value function of the POMDP associated with our problem and in the comparison presented in the simulation results in Section \ref{sec:simulation}. 
It is useful to remark that, by considering a \emph{time-dependent} utility after the decision, in our model the constraints on the detection metrics vary between channels and over time, whereas in the majority of other works the detection metrics constraint are typically constant over time (see \cite{Wald:1947} for analysis of truncated sequential hypothesis). To be more precise, our model can be seen as a Bayesian-risk formulation, where our utility terms can be seen as the Lagrangian multiplier of the constraints associated to the detection metrics, that change over time in light of the \emph{time-dependent} utility function.
It is clear that when $L$ approaches infinity, the strategy developed in this paper would not offer a significant advantage with respect to other simplest quasi-static strategies (as discussed in \cite{javidi_almost_fixed} for a single resource case) or the concatenation of SPRT proposed in \cite{paper_concatenated_SPRT}.
However, for a finite time horizon, the proposed strategy offers a better performance, as discussed both analytically at the end of Section II and shown numerically in Section \ref{sec:simulation}.
Even if the behavior of the decision thresholds can be intuitive, our approach to approximate the decision thresholds (Algorithm \ref{alg:threshold_approximation}) and the bound on the expected sensing time (presented in Lemma \ref{lemma_bound_expected_sensing_time} and used in Algorithm \ref{alg:heuristic}), that form our low-complexity heuristic, are completely novel and, in light of the time-dependence and the absence of a close form expression to map the detection metrics constraints into Lagrange multipliers for our formulation, cannot be mapped directly into previous results. 
The discussion on the asymptotic regret in Section \ref{sec:asympt_regret} is reported for completeness, but the asymptotic analysis is rather straightforward and does not represent the main focus of this work.
%
\section{Problem Formulation}\label{sec:formulation}
\subsection{Problem Statement}\label{subsec:statement}
Assume an agent has $L$ instants of time available for the sensing and exploitation of a set of resources $\N=\{1,2,\dots,N\}$. 
The agent accrues a reward that is a function of an underlying state vector 
\begin{equation}
\bs \triangleq (s_{1}, \ldots, s_N)^T\in \bS\equiv\{0,1\}^N,
\end{equation}
where the entries $s_i\in \{0,1\}$ and $\bS$ is the set of all possible states. 
We consider the $s_i$ as indicators of good ($0$) or bad ($1$) state of a resource. 
For instance,  the ``idle'' or ``busy'' state of a sub-band, where the decision maker has to explore/sense the channels and gets a reward for utilizing an idle channel and a penalty for utilizing a busy channel.
We assume the states $s_i$ are mutually independent Bernoulli random variables with known prior probabilities given by:
\begin{align}
\boldsymbol{\omega}[0]&=\{\omega_i[0]:  i=1,2,\dots,N\}\\
\omega_i[0]&\triangleq P(s_i=0).
\end{align}
Let $\A_k$ denote the set of resources for which a final decision of utilizing or discarding has not been reached at time $k$ ($k=0,1,\ldots, L-1$). 
Clearly, $\A_0$ includes all $N$ resources. 
Our model then allows to access multiple resources (removed from $\A_{k}$) and to sense one resource ($\phi_k$) at each time instant $k$.
Underneath the decision there is a sequential binary hypotheses testing problem where a sample $o[k]$ from the selected resource $\phi_k$ is collected and the conditional probability density functions of the observations are assumed to be known:
\begin{align}
&H_0^{\phi_k}:s_{\phi_k}=0, ~~~ o[k]\overset{i.i.d}{\sim} f_0^{\phi_k}(o)\\
&H_1^{\phi_k}:s_{\phi_k}=1, ~~~ o[k]\overset{i.i.d}{\sim} f_1^{\phi_k}(o) 
\end{align}
We want to maximize a utility function that strikes the best trade-off between the need of acquiring information on the state $\bs$ of the channels and the desire of exploiting good resources as early as possible.  
The decision maker needs to design: 1) a set of $N$ stopping rules for the stopping times $\btau\triangleq\{\tau_i:i=1,2,\ldots,N\}$, one for each resource, indicating when a final decision of utilizing or discarding the resource can be made; 2) a set of $N$ decision rules $\bdelta\triangleq\{\delta_i\in\{0 (\mbox{utilize}), 1 (\mbox{discard})\}:i=1,2,\ldots,N\}$, one for each resource, indicating the final decision at the time $\tau_i$ of stopping; 3) a sequence of selection rules $\bphi\triangleq\{\phi_k:k=0,1,\ldots,L-1\}$ indicating which resource to sense at time $k$ (if $\A_k$ is not empty). 
Let $\tau_i$ denote the time instant at which a final decision on whether to utilize or discard resource $i$ is made and $\tilde{\tau}=\max_i{\tau_i}$ be the total sensing time, all depending on the two sets of rules  $(\btau, \bphi)$.
The decision maker's actions are the solutions to the following optimization problem:
\begin{equation}\label{eq:opt_problem_initial_formulation}
\begin{aligned}
& \underset{\btau,\bdelta,\bphi}{\text{max.}}
& & \!\!\!\!\!\!\!\!\!\!\!\!\!\!\mathds{E}\!\left[U(\bs,\N,L,c,\btau,\bdelta,\bphi)\right]\!\triangleq\!\mathds{E}\!\left[\!-c\tilde{\tau}\!+\!\sum_{i=1}^N (L-\tau_i)R_i(s_i;\delta_i)\!\right]\!\!\!\!\\
& \text{subject to}
& & \tilde{\tau}\leq L-1. \\
\end{aligned}
\end{equation}
where the objective function has two terms:
\begin{itemize}
\item $-c\tilde{\tau}$ represents an effective sensing cost, and $c$ is the sensing cost per unit of time;
\item the function $R_i(s_i;\delta_i)$ is the utility per unit time for exploiting resource $i$:
\be\label{eq:utility_single}
R_i(s_i;\delta_i)=\begin{cases}r_i ~~~&\mbox{if}~~\delta_i=s_i=0\\
-\rho_i ~~~~~&\mbox{if}~~\delta_i=0, s_i=1\\
0~~~~&\mbox{if}~~\delta_i=1\end{cases}
\ee 
where $r_i,\rho_i>0$ indicate, respectively, the reward and the penalty for utilizing a {\it good} resource and utilizing a {\it bad} resource.
\end{itemize}
We can immediately notice an asymmetry in the function corresponding to the {\it exploitation reward} in \eqref{eq:utility_single} since, if one decides the resource is in {\it bad} state (i.e. that $s_i=1$), the utility accrued is the same regardless of the real state of the channel. 

Note that, since $c$ can be moved out of the expectation in \eqref{eq:opt_problem_initial_formulation},
$\mathds{E}\left[U^*(\bs,\N,L,c)\right]=\underset{\tau,\delta,\phi}{\max}~\mathds{E}\left[U(\bs,\N,L,c,\btau,\bdelta,\bphi)\right]$ decreases  monotonically with $c$, therefore the value of $c$ is typically limited by an effective cost (i.e. the energy required for the receiver to continue sensing a channel in a spectrum access problem) and is not chosen to further optimize the achievable expected utility.
\subsection{A POMDP Formulation}\label{subsec:POMDP}
In this subsection, we will model the optimization in \eqref{eq:opt_problem_initial_formulation} as a Partially Observable Markov Decision Problem (POMDP), and use dynamic programming tools to describe the optimal $\btau^*,\bdelta^*,\bphi^*$.
The state vector $\bs$ is not directly observable, therefore the decision maker has to rely on her belief regarding the occupancy of the resources in order to make a decision. 
Let $\boldsymbol{\omega}[k]$ denote the belief vector at instant $k$ (i.e., the vector of posterior probabilities that a resource is in state $0$ given all past observations). 
We can use Bayes rule to derive the belief update after a new observation and write  
\begin{align}
\boldsymbol{\omega}[k+1]&=\bPi(\boldsymbol{\omega}[k],o[k],\phi_k)\label{eq:belief_update}\\
\omega_i[k+1]&=\Pi_i(\omega_i[k],o[k],\phi_k)\\
\!\!\!\Pi_i(\omega_i,o,\phi)&\triangleq\begin{cases}\dfrac{\omega_if_0^i(o)}{\omega_i f_0^{i}(o)+(1-\omega_i)f_1^{i}(o)} ~~&\mbox{if}~~i=\phi\\ \omega_i ~~&\mbox{if}~~i\neq \phi \end{cases}
\end{align}

The optimization problem in \eqref{eq:opt_problem_initial_formulation} can therefore be formulated as a POMDP where the belief vector $\boldsymbol{\omega}[k]$ represents the {\it state} for the decision maker and the state transitions equations are given by the belief update rule in \eqref{eq:belief_update}.
At each time instant $k$, based on the current belief $\bomega[k]$, the decision maker first removes from $\A_k$ all those resources for whom a final decision of utilizing or discarding can be made (we refer to this set as $\D$), and then she chooses one of the remaining resources to sense at time $k$.
The value function of the POMDP problem can be expressed as:
\begin{align} 
&V(\bomega,\A_k,k)\label{eq:value_function_exp}\\
&\triangleq\max\limits_{\D\subseteq \A_k}\left\{(L-k)V_d(\bomega,\D)+\max\limits_{i\in\A_{k+1}\equiv\A_k\setminus\D}V_t^i(\bomega,\A_{k+1},k)\right\}\nonumber
\end{align} 
where $V_d(\bomega,D)$ indicates the maximum expected reward given by the resources in $\D$ for each of the $(L-k)$ remaining instants and $V^i_t(\bomega,\A,k)$ represents the value accrued for deciding to sense the resource $i$ at time $k$ (i.e. $\phi_k=i$).
It is easy to see from the utility function \eqref{eq:utility_single} that $V_d$ has an additive structure, i.e. $V_d(\bomega,\D)=\sum_{i\in\D}V_d^i(\omega_i).$
Mathematically the functions $V_d^i(\omega_i)$ and $V_t^i(\omega_i,\A,k)$ can be expressed as follows:
\begin{align}
V_d^i(\omega_i)&\triangleq\max_{\delta_i\in\{0,1\}}\mathds{E}\left[R_i(s_i;\delta_i)\right]\label{eq:V_d_def}\\
\!\!\!\!\!V_t^i(\boldsymbol{\omega},\A,k)\!&\triangleq\!-c\!+\!\!
\int\! V\!\left(\bPi(\bomega,o,i),\A,k+1\right)f_{1-\omega_i}^i(o) do\label{eq:V_t_def}
\end{align}
for $k=0,1,\dots,L-1$ {where we need a final condition $V(\bomega,\A,L)=0, \forall \A,\bomega$,  to encode the constraint in \eqref{eq:opt_problem_initial_formulation}}. 
In \eqref{eq:V_t_def} we have defined 
\be 
f_{1-\omega_i}^{(i)}(o)\triangleq \omega_i f_0^{(i)}(o)+(1-\omega_i)f_1^{(i)}(o)
\ee 
It is easy to see that
\be 
V^i_d(\omega_i)=\max\{(r_i+\rho_i)\omega_i-\rho_i,0\}.\label{eq:V_d_def_no_delta}
\ee 
Thus the optimal final decision $\delta^*_i$ is given by
\be\label{eq:optimal_decision_rule}
\delta^*_i=u\left(\frac{\rho_i}{\rho_i+r_i}-\omega_i\right)
\ee
where $u(\cdot)$ is the unit step function.
The value function 
$V(\bomega,\A_k,k)$ can be seen as the result of the maximization of a set function under set constraints, that is:
\be\label{eq:set_function_value_function}
V(\bomega,\A_k,k)=\max\limits_{\D \subseteq \A_k}J(\D)
\ee
where
\begin{align} 
\!\!J(\D)\triangleq\left\{(L-k)V_d(\bomega,\D)+\!\!\max\limits_{i\in\A_k\!\setminus\!\D}V_t^i(\bomega,\A_k\!\setminus\!\D,k)\right\}\label{eq:set_function}
\end{align} 
and this formulation will be used in the remainder of this work to show the structure of the optimal policy.
We would like to point out that the formulation of the POMDP in \eqref{eq:value_function_exp} gives no indication on the fact that the channel sensed at time $k$ should continue to be sensed at time $k+1$ or included in the set $\D$, i.e. a concatenation of independent truncated SPRT over each channel represents a suboptimal strategy for our problem and is optimal only for $L\rightarrow\infty$. 
\begin{remark}
Our formulation can be modified to account for correlation between group of resources. If we keep the condition that it is possible to sense only one resource at each time, the belief update rule in \eqref{eq:belief_update} should be modified to update also the $\omega_i$'s of resources correlated with the sensed one. The subset of resources, that can be accessed at each time, should contain all the ones correlated with each other, i.e. if two resources are correlated, they should be added to $\mathcal{D}$ at the same time, due to the possibility of gaining knowledge on a resource for which a terminal decision has already been made. It follows that, for each group of correlated resources, we have to consider a sub multi-hypothesis problem for all the possible combinations of binary states, and find a similar approximation as the one presented in the next section for the decision regions in the plane of the belief vector over these resources (for details on the geometry of such structure see \cite{multi_procedure_baum}). This in principle can be handled for small groups of correlated resources, but makes the problem even more complex.
\end{remark}
%
\section{Opportunistic Sequential Sensing}\label{sec:design}
\subsection{The Optimal Stopping Rule and Decision Rule}
In this section we first describe the correspondence between the actions of our decision process as solution of \eqref{eq:value_function_exp} and the optimal rules $(\btau^*,\bphi^*)$ introduced in \ref{subsec:statement}, that the decision maker needs to determine.
Then we introduce a low-complexity policy to approximate the optimal action.
Let us start by considering the solution of \eqref{eq:set_function_value_function} at time $k$, i.e.:
\be 
\D^*=\argmax\limits_{\D\subseteq \A_k}J(\D).
\ee
By definition of the set of optimal stopping times $\btau$ (see \eqref{eq:opt_problem_initial_formulation}), for each resource $i$, the optimal stopping time $\tau^*_i$ is:
\be\label{eq:optimal_stopping_rule}
\tau^*_i=\begin{cases}k,~&\mbox{for}~~i\in\D^*\\
k'>k~~&\mbox{for}~~i\in\A_k\setminus\D^*.
\end{cases}
\ee
In light of the outer maximization over $\D$ in \eqref{eq:value_function_exp}, the optimal selection rule for the decision maker will be: 
\be\label{eq:optimal_selection_rule}
\phi^*_k=\argmax_{i\in\A_{k+1}\equiv\A_k\setminus\D^*}V_t^i(\bomega,\A_{k+1},k)
\ee
From \eqref{eq:optimal_stopping_rule} and \eqref{eq:optimal_selection_rule} it emerges how the optimal selection rules $\bphi^*$ and the optimal stopping times $\btau^*$ are coupled.
We then introduce the following lemma:
\begin{lemma}\label{lemma_convexity_V_i}
{\it
$\forall { \A_k}\subseteq\N$, $i,j\in\A_k$ and for $k=0,\dots,L-1$, both
$V^i_t(\bomega,\A_{k+1},k)$, $V(\bomega,\A_k,k)$ are convex functions of $\omega_j$.}
\end{lemma}
\begin{proof}
The proof is in Appendix \ref{app:proof_convexity_V_i}. It is similar to the approach used in \cite{Zhao&Ye:2010TSP} to prove the convexity of the Q-functions.  \end{proof}
Lemma \ref{lemma_convexity_V_i} induces the structure of the optimal stopping rules, which describe the optimal stopping times $\tau^*_i,~i=1,2,\ldots,N$ that we formally present in the following theorem:
\begin{theorem}\label{th:optimal_policy_structure}
{\it
The optimal stopping time $\tau^*_i, \forall i\in\N$ is described by two thresholds $(\nu_1^{i},\nu_0^{i})$ that depend on the remaining time and channels to be explored, i.e. ($\nu_1^{i}=\nu_1^{i}(\A_k,k),\nu_0^{i}=\nu_0^{i}(\A_k,k)$) such that at any $k$ the optimum action for the resource $i$ is
\begin{itemize}
\item take a final decision over resource $i$ (i.e. $\tau_i=k$) if \\$\omega_i[k]\geq\nu_0^{i}(\A_k,k)~\vee~\omega_i[k]\leq\nu_1^{i}(\A_k,k)$ 
\item postponing the decision on resource $i$ (i.e. $\tau_i>k$), if \\$\nu_1^{i}(\A_k,k)<\omega_i[k]<\nu_0^{i}(\A_k,k)$
\end{itemize}
Furthermore $\forall i,\A_k,k$ \\
$\nu_1^{i}(\A_k,k)\leq\frac{\rho_i}{\rho_i+r_i}\leq\nu_0^{i}(\A_k,k)$ and from \eqref{eq:optimal_decision_rule} it follows that: 
\be
\delta_i^*=\begin{cases}1 &\mbox{if}~\omega_i[k]\leq\nu_1^i(\A_k,k)\\ 
0 &\mbox{if}~\omega_i[k]\geq\nu_0^i(\A_k,k).\end{cases}
\ee }
\end{theorem}
\begin{proof}
Following the same approach we used to show the convexity of $V(\bomega,\A_k,k)$ in Appendix \ref{app:proof_convexity_V_i}, we rewrite \eqref{eq:value_function_exp} for an arbitrary resource $i$ as follows:
\be\label{eq:rewrite_value_function}
\!\!V(\bomega,\A_k,k)\!=\!\max\left\{\max\limits_{\{i\}\subseteq\D\subseteq\A_k}J(\D),\max_{\D\subseteq\A_k\setminus\{i\}}J(\D)\right\}
\ee
where the function $J(\D)$ has been defined in \eqref{eq:set_function}, and then we refer to the two terms of the outer  maximization as $f_1=\max\limits_{\{i\}\subseteq\D\subseteq\A_k}J(\D)$ and $f_2=\max_{\D\subseteq\A_k\setminus\{i\}}J(\D)$.
Note that $f_1$ corresponds to the maximum value from the actions that immediately decide on resource $i$, whereas $f_2$ indicates the maximum value from the actions that do not decide on resource $i$ at time $k$.
The convexity of $f_1$ and $f_2$ is proven in Appendix \ref{app:proof_convexity_V_i}. 
Clearly $f_1\geq f_2$ for $\omega_i=0$ and for  $\omega_i=1$ since, when the state of the channel is known, is clearly preferable (or equivalent) to immediately decide on it, and also $f_1$ is a piece-wise linear function of $\omega_i$ with only two segments that intersect in $\omega_i=\frac{\rho_i}{\rho_i+r_i}$ (see definition of $V_d^i(\omega_i)$ in \eqref{eq:V_d_def_no_delta}).
Hence, there are  only two possibilities:
\begin{itemize}
\item  $f_1<f_2$ for $\nu_1^{i}(\A_k,k)<\omega_i[k]<\nu_0^{i}(\A_k,k)$ where \\$0\leq\nu_1^{i}(\A_k,k)\leq\frac{\rho_i}{\rho_i+r_i}\leq\nu_0^{i}(\A_k,k)$;\vspace{0.2cm}
\item $f_1\geq f_2$ $\forall \omega_i\in [0,1]$.
\end{itemize}
In this second case, we will say the two functions do not intersect, i.e. there is no region where the decision maker should prefer to not include the resource in $\D$; therefore we set $\nu_1^{i}(\A_k,k)=\nu_0^{i}(\A_k,k)=\frac{\rho_i}{\rho_i+r_i}$ to indicate the decision maker should immediately decide on resource $i$ and this concludes the proof.
\end{proof}
\begin{figure}[ht]
\begin{center}
\includegraphics[scale=0.35]{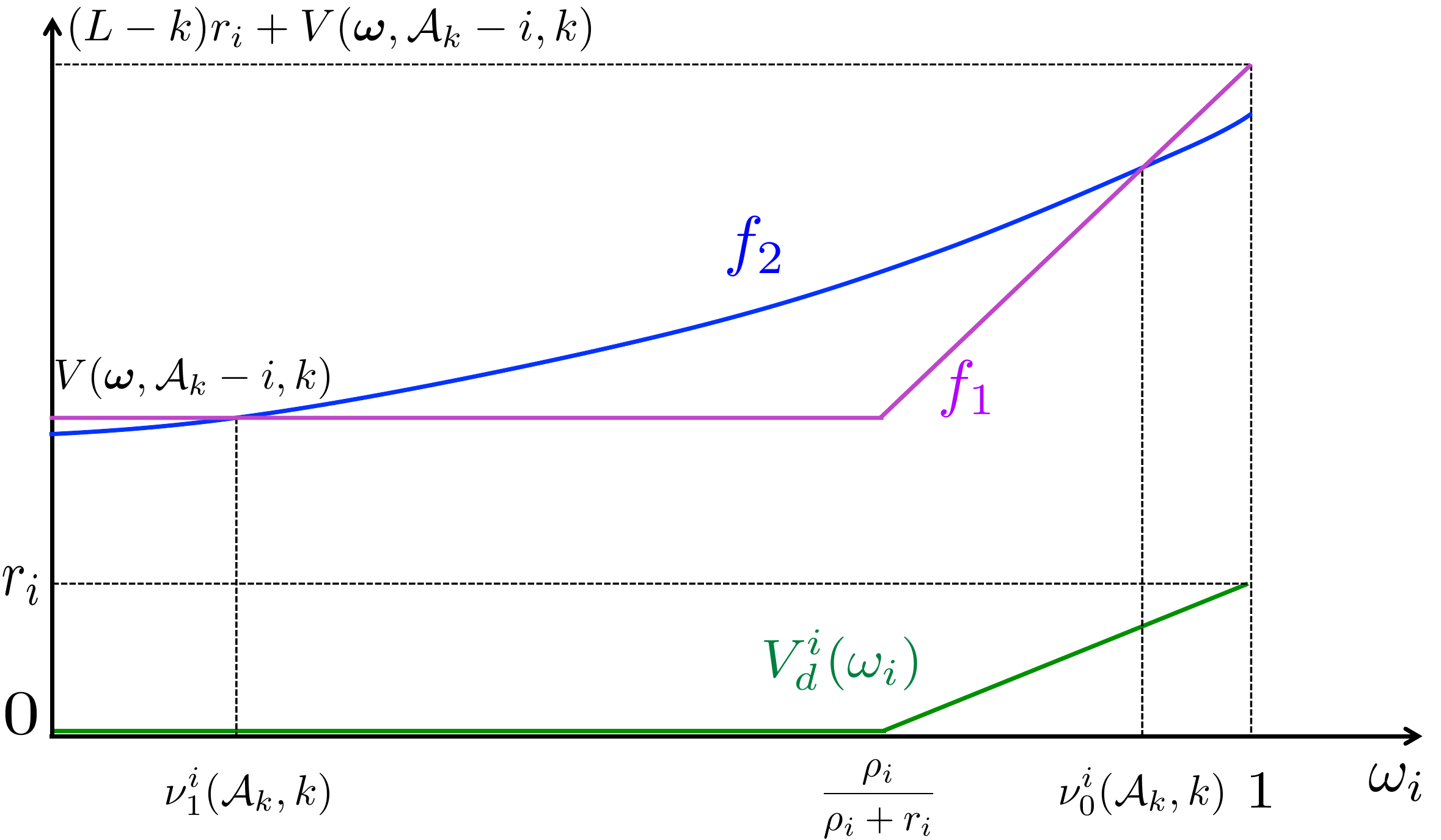}
\caption{Representation of the functions $f_1$ and $f_2$ of $\omega_i$ from \eqref{eq:rewrite_value_function} at time $k$}\label{fig:convexity_proof}
\end{center}
\end{figure}
It is not entirely surprising that the problem we defined leads to an optimal policy with a two thresholds structure, since analogous policies have been found to be optimal for a general truncated SPRT \cite{wald1948} with a deadline constraint either deterministic \cite{book_seq_hyp_test} or stochastic \cite{Frazier08sequentialhypothesis} and in general many different sequential decisions schemes.
However, as described in the Introduction, our formulation is different in light of the time dependence of the Bayesian term, i.e. the different utility that a decision over a certain resource can produce at different times.
\par Theorem \ref{th:optimal_policy_structure} is a description of the optimal stopping rules but does not indicate how to compute the two thresholds nor how to select the resource $\phi^*_k$ to be sensed at time $k$.
Due to the recursive nature of the function $V_t^i(\boldsymbol{\omega},\A_{k+1},k)$ in \eqref{eq:V_t_def} and the dependence on the rest of the system, the exact computation of these optimal thresholds remains elusive.
In order to provide a suboptimal strategy that is computationally manageable in Section \ref{subsec:heuristic}, let us look at the situation where we have only one resource and we have to choose between testing the resource or taking a final decision, either $\delta_i=1$ or $\delta_i=0$. 
We will refer to the two single-resource decision thresholds as $\nu_1^i[k]=\nu_1^i(i,k),~\nu_0^i[k]=\nu_0^i(i,k)$, where we introduce this short notation for convenience.
We then introduce the following lemmas.
\begin{lemma}\label{lemma_time_variant_threshold}
$\forall i\in\N,~k=0,1,\dots,L-2$ the two single-resource thresholds will monotonically contract, i.e. 
\begin{align}
\nu_1^i[k]&\leq\nu_1^i[k+1]\label{eq:nu_1_k}\\
\nu_0^i[k]&\geq\nu_0^i[k+1]\label{eq:nu_0_k}
\end{align}
\end{lemma}
\begin{proof}
The intuition behind this lemma is that, as the decision deadline approaches, the urge to decide whether to utilize or discard the resource increases and the decision thresholds shrink. For a rigorous proof see Appendix \ref{app:time_variant_threshold}.
\end{proof}
\begin{lemma}\label{lemma_threshold_inequalities}
$\forall i\in\N,~k=0,1,\dots,L-1,~\A'\in 2^{\N-i}$, the following inequalities hold 
\begin{align}
\nu_1^i[k]&\leq\nu_1^i(\A'+i,k)
\label{eq:threshold_inequality_1}\\
\nu_0^i[k]&\geq\nu_0^i(\A'+i,k)\label{eq:threshold_inequality_0}
\end{align}
\end{lemma}
\begin{proof}
The intuition behind this lemma is that adding resources to the state $\A$ produces a similar effect to removing time dedicated to each resource. Thus a similar intuition as for Lemma \ref{lemma_time_variant_threshold} applies. The full proof is in Appendix \ref{app:threshold_inequalities}.
\end{proof}
\begin{lemma}\label{lemma_threshold_heuristic}
$\forall k,k'=0,\dots,L-1,~k'\geq k,~\A\in 2^{\N},i\in\A$
\begin{align}
{\nu}_1^i[k]\leq\nu_1^i(\A,k')\label{eq:threshold_inequality_nu_1}\\
{\nu}_0^i[k]\geq\nu_0^i(\A,k')\label{eq:threshold_inequality_nu_0}
\end{align}
\end{lemma}
\begin{proof}
Follows from Lemmas \ref{lemma_time_variant_threshold}-\ref{lemma_threshold_inequalities} and setting $\A=\A'+i$.
\end{proof}
Lemma \ref{lemma_threshold_heuristic} gives us insight on the behavior of the thresholds $\left(\nu_1^{i}(\A_k,k),\nu_0^{i}(\A_k,k)\right)$ that motivates our heuristic strategy in Section \ref{subsec:heuristic}, which uses a single-channel thresholds approximation to replace the optimal thresholds.
In the next section we will first discuss how to approximate the thresholds $\nu_1^i[k]$, $\nu_0^i[k]$ and then we will introduce a manageable strategy to sense the resources and decide which action to take based on the approximate decision thresholds.

\subsection{A Low Complexity Approximation Algorithm}\label{subsec:heuristic}
For the rest of this subsection, we will use underline to indicate a lower bound and overline for upper bound of a given quantity or function, i.e. $\underline{a}\leq a \leq \overline{a}$. 
For our heuristic, as we will motivate in the following, we will use the bounds $\underline{\nu_1^i}[k]$ and $\overline{\nu_0^i}[k]$ that correspond to the single-resource decision thresholds, derived from the formulation previously introduced. 
In principle, the tighter are the bounds we can provide for the thresholds, the better our heuristic will perform.
Let us start by finding simple bounds.
The decision thresholds $(\nu_1^i[k],\nu_0^i[k])$ at  time $k$  are the solutions with respect to $\omega_i$ of 
\begin{equation}\label{eq:intersection}
(L-k)V_d^i(\omega_i)=V_t^i(\omega_i,i,k).
\end{equation}   
From the convexity of $V_t^i(\omega_i,i,k)$ (see Lemma \ref{lemma_convexity_V_i}) and the values of the function for $\omega_i=0,1$ we know that (see Fig.\ref{fig:value_single})
\begin{equation}\label{eq:upper_bound_V_t_convexity}
V_t^i(\omega_i,i,k)\leq -c+(L-k-1)r_i\omega_i. 
\end{equation}
\begin{figure}[ht]
\begin{center}
\includegraphics[scale=0.35]{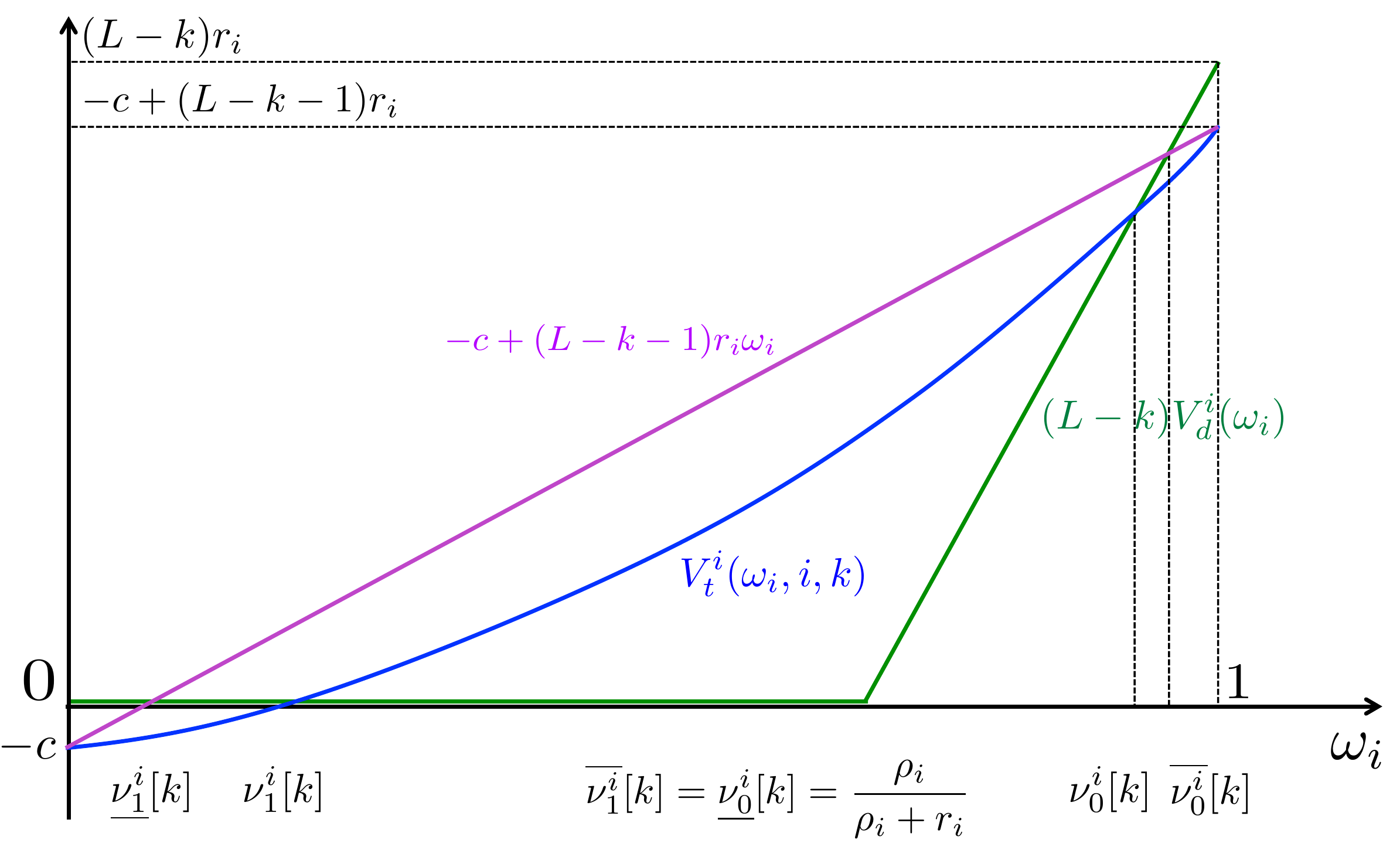}
\caption{Representation of the two functions of $\omega_i$:  $V_t^i(\bomega,i,k)$ associated to the test of resource $i$ and $(L-k)V_d^i(\omega_i)$ for the immediate decision on resource $i$.}\label{fig:value_single}
\end{center}
\vspace{-0.3cm}
\end{figure}
Furthermore, $V_d^i(\omega_i)$ is piece-wise linear with only two segments that intersect in $\omega_i=\frac{\rho_i}{\rho_i+r_i}$. 
\begin{remark}
{ {\it A more general formulation with $4$ different rewards/penalties for \eqref{eq:utility_single} would not alter the structure of the problem nor invalidate our results.
Our definition has been motivated by the emphasis we want to put on the utilization of the resources, which we assume occurs only if the detected state is $0$, but if one had defined 
\begin{equation}\label{eq:rif}
R_i(s_i;\delta_i)\triangleq\begin{cases}C^i_{00}>0&\mbox{if}~s_i=0,\delta_i=0\\
C^i_{10}<0&\mbox{if}~s_i=0,\delta_i=1\\
C^i_{11}>0&\mbox{if}~s_i=1,\delta_i=1\\
C^i_{01}<0&\mbox{if}~s_i=1,\delta_i=0\end{cases}
\end{equation} 
then with some manipulations, one can find the strategy would have the same structure, where Fig.\ref{fig:value_single} would be ``rotated".}}
\end{remark}

By intersecting the upper bound for $V_t^i(\omega_i,i,k)$ in \eqref{eq:upper_bound_V_t_convexity}
with $(L-k)V_d^i(\omega_i)$ it follows:
\begin{align}
\underline{\nu_1^i}[k]&=\min\left\{\dfrac{c}{(L-k-1)r_i},\dfrac{\rho_i}{\rho_i+r_i}\right\}\label{eq:easy_lb_nu_1}\\ 
\overline{\nu_0^i}[k]&=\max\left\{\dfrac{(L-k)\rho_i-c}{(L-k)\rho_i+r_i},\dfrac{\rho_i}{\rho_i+r_i}\right\}\label{eq:easy_ub_nu_0}\\
\overline{\nu_1^i}[k]&=\underline{{\nu_0^i}}[k]=\dfrac{\rho_i}{\rho_i+r_i}\label{eq:easy_central_bound_nu}
\end{align}
where we consider $\nu_1^{i}[k]=\nu_0^{i}[k]=\frac{\rho_i}{\rho_i+r_i}$ if \eqref{eq:intersection} has no solutions for $\omega_i\in[0,1]$.
The bounds above can be very loose depending on the parameters of our problem and lead to a poor approximation.
Potentially tighter upper and lower bounds than the ones directly obtainable from \eqref{eq:easy_lb_nu_1}-\eqref{eq:easy_ub_nu_0}-\eqref{eq:easy_central_bound_nu} can be found using a probabilistic approach to write the value function of our dynamic program.
The general idea is to consider a worst and a best case scenario any time the belief over a certain resource has not crossed one of the two thresholds. 
This is possible since the belief update (for a given observation) and the value function are monotonically increasing functions of $\omega_i$ (details are reported in Appendix \ref{app:threshold_approximation}).
Let us define the following functions for $s=0,1$
\begin{align}
\!\!\!\bar{F}_{\omega_i[k+\ell]}^{\ell}(\omega|\varphi,s)&\!\triangleq\!P\!\left\{\omega_i[k+\ell]\geq\omega|\omega_i[k]\!=\!\varphi\!\cap\!s_i\!=\!s\right\}\label{eq:comple_CDF}\\
\overline{\mu_i}[k|s]&\triangleq \nonumber\bar{F}_{\omega_i[k+1]}^{1}(\underline{\nu_1^i}[k]|\overline{\nu_0^i}[k-1],s)\\&-\bar{F}_{\omega_i[k+1]}^{1}(\overline{\nu_0^i}[k]|\underline{\nu_1^i}[k-1],s)\label{eq:mu_upp}\\
\underline{\mu_i}[k|s]&\triangleq\nonumber\bar{F}_{\omega_i[k+1]}^{1}(\overline{\nu_1^i}[k]|\underline{\nu_1^i}[k-1],s)\\&-\bar{F}_{\omega_i[k+1]}^{1}(\underline{\nu_0^i}[k]|\overline{\nu_0^i}[k-1],s)\label{eq:mu_und}
\end{align}
where \eqref{eq:comple_CDF} is the complementary CDF of the updated belief, which is a random variable while \eqref{eq:mu_upp}-\eqref{eq:mu_und} refer to the probability of not overcoming the belief thresholds.

We will then prove
\begin{lemma}\label{lemma_threshold_approximation}
{\it $\forall i\in\N,k=0,1,\dots,L-1$ we can derive the following upper and lower bound for the function $V_t^i(\omega_i,i,k)$
\begin{align}
&\overline{V_t^i}(\omega_i,i,k)=-c\nonumber\\
&+\min\left\{\sum_{\ell=k}^{L-2}\left[\omega_i\hspace{-0.2cm}\prod\limits_{m=k+1}^{\ell}\hspace{-0.2cm}\overline{\mu_i}[m|0]+(1-\omega_i)\hspace{-0.2cm}\prod\limits_{m=k+1}^{\ell}\hspace{-0.2cm}\overline{\mu_i}[m|1]\right]\right.\nonumber\\
&\max\hspace{-0.1cm}\left\{\hspace{-0.1cm}-u[\ell-k]c+(L-\ell-1)\left[\overline{\nu_0^i}[\ell] r_i \bar{F}^1_i\hspace{-0.1cm}\left(\underline{\nu_0^i}[\ell+1]|\overline{\nu_0^i}[\ell],0\right)\right.\right.
\nonumber\\
&\left.\left.\left.\hspace{-0.1cm}-(1\hspace{-0.1cm}-\overline{\nu_0^i}[\ell])\rho_i\bar{F}_i^1\hspace{-0.1cm}\left(\overline{\nu_0^i}[\ell+1]|\underline{\nu_1^i}[\ell],1\right)\right],0\right\},(L-k-1)r_i\omega_i\right\}\label{eq:upper_bound_V_t}
\end{align}
\begin{align}
&\underline{V_t^i}(\omega_i,i,k)=\hspace{-0.1cm}-c+\hspace{-0.1cm}\sum_{\ell=k}^{L-2}\left[\omega_i\hspace{-0.3cm}\prod\limits_{m=k+1}^{\ell}\hspace{-0.3cm}\underline{\mu_i}[m|0]+(1-\omega_i)\hspace{-0.3cm}\prod\limits_{m=k+1}^{\ell}\hspace{-0.3cm}\underline{\mu_i}[m|1]\right]\nonumber\\
&\max\hspace{-0.1cm}\left\{\hspace{-0.1cm}-u[\ell-k]c+(L-\ell-1)\left[\underline{\nu_1^i}[\ell] r_i \bar{F}^1_i\hspace{-0.1cm}\left(\overline{\nu_0^i}[\ell+1]|\underline{\nu_1^i}[\ell],0\right)\right.\right.\nonumber\\
&\left.-(1-\underline{\nu_1^i}[\ell])\rho_i\bar{F}_i^1\left(\underline{\nu_0^i}[\ell+1]|\overline{\nu_0^i}[\ell],1\right),0\right\}\label{eq:lower_bound_V_t}
\end{align}
where we used the step function $u[k]=1$ for $k>0$ and $u[0]=0$ and the short notation $\bar{F}_i^1(\omega|\varphi,s)$ for $\bar{F}_{\omega_i[\ell+1]}^1(\omega|\varphi,s)$}.~\\
\end{lemma}

\begin{proof}
See Appendix \ref{app:threshold_approximation}. Notice that, since we evaluate these bounds at time $k$ in $\omega_i$, the bounds for the thresholds at time $\ell=k$ can be replaced with $\omega_i$ (same holds for the bound in \eqref{eq:upper_bound_V_t_marketing} in Section \ref{sec:new_application}).
\end{proof}
Gathering these results, we can now present our thresholds approximation in Algorithm \ref{alg:threshold_approximation}. Note that our approach is valid irrespective of the sensing model. In fact, in order to express in closed form the bounds in \eqref{eq:upper_bound_V_t}-\eqref{eq:lower_bound_V_t} one only needs to evaluate the functions $\bar{F}^1_{\omega_i[k+1]}(\omega|\varphi,s)$ in \eqref{eq:comple_CDF}.
 
We highlight it is important to run the algorithm for decreasing value of $k$ in order to use potentially tighter bounds when we compute \eqref{eq:upper_bound_V_t}-\eqref{eq:lower_bound_V_t} for lower values of $k$.
The last step in the algorithm is a direct consequence of Lemma \ref{lemma_time_variant_threshold}.
In Section \ref{sec:simulation} we will numerically illustrate how the tightness of the bounds from Algorithm \ref{alg:threshold_approximation} is highly dependent on the value of the parameters of our utility function and on the quality of the test (i.e the sensing $SNR$ $\zeta_i$ for our application).
\begin{algorithm}[ht]
\SetAlgoLined~\\
$\forall i\in\N$\\
 1) set $\underline{\nu_1^i}[L-1]=\overline{\nu_1^i}[L-1]=\underline{\nu_0^i}[L-1]=\overline{\nu_0^i}[L-1]=\dfrac{\rho_i}{\rho_i+r_i}$\\~\\
2)~\For{$k=L-2:-1:0~~$} 
    {\begin{itemize}
    \item Compute $\overline{V_t^i}(\omega_i,i,k)$ from \eqref{eq:upper_bound_V_t}          and $\underline{V_t^i}(\omega_i,i,k)$\\ from \eqref{eq:lower_bound_V_t}
    \item  find the numerical solutions $\omega_1<\omega_2$ of 
\begin{equation*}
\overline{V_t^i}(\omega_i,i,k)=(L-k)V_d^i(\omega_i)
\end{equation*}
where $\omega_1=\omega_2=\frac{\rho_i}{\rho_i+r_i}$ if no solutions.\\
Set 
\begin{align*}
\underline{\nu_1^i}[k]&=\omega_1\\
\overline{\nu_0^i}[k]&=\omega_2
\end{align*}
    \item  find the two numerical solutions $\omega_1<\omega_2$ of 
\begin{equation*}
\underline{V_t^i}(\omega_i,i,k)=(L-k)V_d^i(\omega_i)
\end{equation*}
where $\omega_1=\omega_2=\frac{\rho_i}{\rho_i+r_i}$ if no solutions.\\
 Set 
 \begin{align*}
 \underline{\nu_1^i}[k]&=\min(\omega_1,\underline{\nu_1^i}[k+1])\\
 \overline{\nu_0^i}[k]&=\max(\omega_2,\overline{\nu_0^i}[k+1])
 \end{align*}
    \end{itemize} }
\caption{Thresholds Approximation Algorithm}\label{alg:threshold_approximation}
\end{algorithm}
Having introduced two methods to approximate the thresholds, we now introduce our low-complexity strategy.
The pseudocode for our heuristic is presented in Algorithm \ref{alg:heuristic}.
{Interestingly, the distance between the two bounds can be used to find an on-line (not computable {\it a priori}) non trivial upper-bound for the cumulative utility loss of Algorithm 2.
}
\begin{lemma}\label{lemma_bound_expected_sensing_time}
{\it The expected sensing time for resource $i$, given we continue to sense it from time $k$ until we make a final decision at time $\tau_i$, can be upper-bounded as follows:
\begin{align}
&\mathds{E}[\tau_i-k|k]\leq\min\left\{\omega_i[k]\frac{\varsigma\left(\overline{\nu_0^i}[k],\omega_i[k]\right)+\hat{D}(f_0^i||f_1^i)}{D(f_0^i||f_1^i)}\right.\nonumber\\&\left.+(1-\omega_i[k])\frac{\varsigma\left(\omega_i[k],\underline{\nu_1^i}[k]\right)+\hat{D}(f_1^i||f_0^i)}{D(f_1^i||f_0^i)},L-k-1\right\}\label{eq:E_tau_i}
\end{align}
where $\varsigma(x,y)\triangleq\log\left(\frac{x(1-y)}{(1-x)y}\right)$ and   
\begin{align}
\hat{D}(f_0^i||f_1^i)&\triangleq\mathds{E}\left\{\log\left(\frac{f_0^i(o)}{f_1^i(o)}\right)\big|s_i=0,{f_1^i(o)}\leq{f_0^i(o)}\right\}\label{eq:def_hat_div_01}\\
\hat{D}(f_1^i||f_0^i)&\triangleq\mathds{E}\left\{\log\left(\frac{f_1^i(o)}{f_0^i(o)}\right)\big|s_i=1,{f_1^i(o)}\geq{f_0^i(o)}\right\}\label{eq:def_hat_div_10}
\end{align}
}
\end{lemma}
\begin{proof}
See Appendix \ref{app:proof_bound_expected_sensing_time}
\end{proof}
The reason for immediately deciding on the resources whose belief has crossed one of the two single-channel thresholds  follows from Lemma \ref{lemma_threshold_heuristic}, which implies that at any time instant $k'$ after $k$ the decision maker will always decide not to test resource $i$ adding $kV_d^i(\omega_i)$ to the overall utility.
Therefore postponing that decision after $k$ will give a lower or equal expected utility. 
In other words, if $V_d^i(\omega_i)>0$ and we expect to select to exploit resource $i$ at some point, we will have a lower utility postponing that decision; if, instead, $V_d^i(\omega_i)=0$ and we discard channel $i$ we accrue zero  utility irrespective on when we take that action.
Using larger thresholds than the optimal ones appears a legitimate choice for applications where the accuracy of the test has a larger impact on the performances than the time devoted to sensing.

We then approximate the optimal selection rule $\bphi^*$ with 
the index $\frac{\omega_i[k]r_i}{\mathds{E}[\tau_i-k|k]}$ which is computed by using the upper-bound introduced in Lemma \ref{lemma_bound_expected_sensing_time}.
The motivation for the index and its asymptotic optimality are discussed in Appendix \ref{app:index_motivation}.

\begin{algorithm}[ht]
\SetAlgoLined
 $k=0$,
 $\A_k=\N$\;
 \While{$k<L$ \textbf{and} $\A_k\neq\emptyset$}
 {  
 $\D=\emptyset$\;
  1) Search for channels $i$ such that $\omega_i[k]<\underline{{\nu}_1^{i}}[k] \vee \omega_i[k]>\overline{{\nu}_0^{i}}[k]$\;  
  \If{there are such channels}
  {
  \For{every channel $i$ that has been found}
   {
    $\D=\D\cup\{i\}$\;     
   }
   }
   2) Take a decision over the resources in $\D$, accrue utility $(L-k)V_d(\bomega,\D)$ and remove $\D$ from the state ($\A_{k}-\D\rightarrow\A_{k+1}$).\\
  3) Test channel 
 \begin{equation}\label{eq:index_sorting}
 \phi_k=\argmax\limits_{i\in\A_{k+1}}\dfrac{\omega_i[k]r_i}{\mathds{E}[\tau_i{-k}|k]}
 \end{equation}
 and update $\bomega[k]$ from \eqref{eq:belief_update}.\\
  4) $k\rightarrow k+1$; 
 }
 \caption{Heuristic for the joint design of $\btau$ and $\bphi$}\label{alg:heuristic}
\end{algorithm}

\subsection{Asymptotic regret}\label{sec:asympt_regret}
Even if this work focuses on finite horizon, it is useful to briefly discuss  the asymptotic regret. 
Let us define the regret $\Psi(L)$ as the \emph{cumulative} loss over the entire horizon $L$
\begin{equation}\label{eq:regret}
\Psi(L)=L\sum\limits_{i\in\N}\omega_i[0]r_i-\mathds{E}\left[U(\bs,\N,L,c)\right]
\end{equation}
where $L\sum_{i\in\N}\omega_i[0]r_i$ represents the maximum achievable expected utility by a genie that knows exactly the state of each resource.
To find an asymptotically optimal strategy in term of regret per slot, i.e. a strategy that achieves 
$\lim_{L\rightarrow\infty}\frac{\Psi(L)}{L}=0$, 
one can simply consider a strategy that senses for a fixed amount of time growing with $\log L$, that would asymptotically achieve zero probability to declare the wrong status for each resource, and then directly derive the limit. 
This proposed static strategy would give a \emph{cumulative} regret with order $O\left({\log{L}}\right)$.
In our simulation, we will then consider, for the different proposed strategies, the growth rate of the regret with $\log L$, specifically 
$\lim_{L\rightarrow\infty}\frac{\Psi(L)}{\log{L}}$ to highlight the importance of the resource sorting in reducing this quantity.
The static approach, however, is clearly not feasible in the regime of small $L$ and this motivates our study of the optimum  policy structure and our heuristic for such scenarios.
\section{Applications}\label{sec:applications}
\subsection{Sequential Spectrum Sensing}\label{sec:sensing_model}
We now use our sequential decision model to find a single-band spectrum sensing technique for cognitive radio systems.  
The decision maker is a secondary transmitter and each resource $i$ is associated with a ``channel'' to communicate with a secondary receiver. 
The term channel is used in a broad sense and it could represent a frequency band with certain bandwidth, a set of tones in an OFDM system or a collection of spreading codes in a Code Division Multiple Access (CDMA) system.
We denote with $\gamma_i^{P,S}$ the two instantaneous Signal-to-Noise Ratios at the receiver for the primary and secondary communication respectively. 
$\xi_i$ represents the instantaneous $SNR$ from the secondary transmitter that could interfere with the primary communication and $\zeta_i$ indicates the instantaneous $SNR$ at the secondary transmitter for the primary transmitter, which is the information the secondary transmitter uses to detect the primary communication.
We consider all the communication channels of interest as Rayleigh fading channels, therefore all the $SNR$'s are exponentially distributed, i.e.
\be\label{eq:pdf_choice}
p_{\eta}(x)=\dfrac{1}{\bar{\eta}}\exp\left(-\dfrac{x}{\bar{\eta}}\right)~~\mbox{for}~\eta\in\left\{\gamma_i^P,\gamma_i^S,\xi_i,\zeta_i\right\}
\ee 
where $\bar{\eta}$ indicates the average value of $\eta$. The phase of the received signal is uniform.
If the secondary transmitter decides to transmit over the $i$-th channel ($\delta_i=0$) and the Primary Transmitter is not transmitting ($s_i=0$) then it accrues a utility per time instant which is a function of ${\gamma}_i^S$.
For our model, we decide to choose the outage rate ($R^S_{i,out}$) as the reward per time instant left ($r_i$), since it is reasonable to assume a slowly varying channel and that the instantaneous CSI (Channel Side Information) is not known at the receiver, therefore the secondary transmitter uses a constant data rate to transmit.
We refer to a design parameter $P^S_{i,out}$ that indicates the probability that the system can be in outage, i.e. the probability that the secondary transmitter cannot successfully decode the transmitted symbols.
Since we assume Rayleigh fading, we can express our reward per time instant left
\be\label{eq:r_i_def}
r_i=\left(1-P^S_{i,out}\right)W_0\log_2\left[1-\bar{\gamma}_i^S\ln\left(1-P^S_{i,out}\right)\right]
\ee
where the derivation can be found in \cite{outage_capacity_def}. If instead, the secondary transmitter decides to use the channel ($\delta_i=0$), interfering with the Primary Transmitter communication ($s_i=1$), then it receives a penalty equal to the loss in outage rate caused to the Primary User, considering its interference adds to the noise at the Primary Receiver.
Let us in fact assume that the Primary Transmitter also transmits at a certain outage rate, given by its value of $P^P_{i,out}$. 
In presence of interference from the secondary user its effective transmission rate becomes
\be\label{eq:R_i_xi_P}
R^P_{i,\xi_i}=(1-P^P_{i,\xi_i})W_0\log_2\left[1-\bar{\gamma}_i^P\ln\left(1-P^P_{i,out}\right)\right]
\ee 
where 
\be\label{eq:deriv_P_i_bar_xi_i}
P^P_{i,\bar{\xi}_i}=\dfrac{P_{i,out}^P-\bar{\xi}_i\ln\left(1-P_{i,out}^P\right)}{1-\bar{\xi}_i\ln\left(1-P_{i,out}^P\right)}
\ee
(for details see Appendix \ref{app:derivation_R_i_xi_P}).
The penalty $\rho_i$ is therefore defined as the loss in rate caused by the interference
\begin{align}
\rho_i&=R_{i,out}^P-R^P_{i,\xi_i}\nonumber\\&=\left(P^P_{i,\bar{\xi}_i}-P^P_{i,out}\right)W_0\log_2\left[1-\bar{\gamma}_i^P\ln\left(1-P^P_{i,out}\right)\right]\label{eq:rho_i_def}
\end{align}
In our model, we assume that cross-channel interference is negligible, i.e. the only possible interference for the primary communication over channel $i$ is given by a secondary transmission over the same channel.
For this work, we also assume that the decision maker has perfect knowledge of the averages $\bar{\gamma}_i^P,\bar{\gamma}_i^S,\bar{\xi}_i,\bar{\zeta}_i$ and the designed $P_{i,out}^P, P_{i,out}^S$ for all the $N$ resources.
%
{ Our signal model considers standard AWGN, therefore when no primary transmission is present, the received baseband complex equivalent signal (only noise) samples are  independent identically distributed $\mathcal{CN}(0,N_0)$ (which stands for complex circularly symmetric zero mean normal Gaussian samples with variance $N_0$) whereas, when a primary transmission exists, since we have assumed Rayleigh fading, the received signal samples are $\mathcal{CN}(0,P_i+N_0)$, where $P_i$ already accounts for the path-loss which we consider deterministic (see Section \ref{sec:simulation_cognitive} for details on the PHY layer). It follows that a sufficient statistic for the detection problem is given by the absolute value squares of the samples of the received signal, which are $\chi^2_2$-distributed, i.e. exponential.
We then consider $f_0^i(o),f_1^i(o)$ exponential with parameter $\theta_0^i,\theta_1^i$ and the ratio $\frac{\theta_1^i}{\theta_0^i}=1+\bar{\zeta}_i$.}

For this particular sensing model, the KL distances are known and the two terms $\hat{D}(f_0^i||f_1^i), \hat{D}(f_1^i||f_0^i)$ can be computed directly from their definition in \eqref{eq:def_hat_div_01}-\eqref{eq:def_hat_div_10}: 
\begin{align*}
\hat{D}(f_0^i||f_1^i)&=\dfrac{\log(1+\bar{\zeta}_i)}{1-\left(1+\bar{\zeta}\right)^{-\frac{\bar{\zeta}_i+1}{\bar{\zeta}_i}}}-\dfrac{\bar{\zeta}_i}{1+\bar{\zeta}_i},~~\hat{D}(f_1^i||f_0^i)=\bar{\zeta}_i.
\end{align*}
We can then directly evaluate the upper-bound in \eqref{eq:E_tau_i} for $\mathds{E}[\tau_i-k|k]$.
We can also give a closed form expression for the complementary $CDF$ $\bar{F}_{\omega_i[k+\ell]}^{\ell}(\omega|\varphi,s)$ in \eqref{eq:comple_CDF}, that will be used for different results that we will present in the next section. 
For our specific sensing model, by simple algebra and use the belief update equation in \eqref{eq:belief_update} we can prove that:
\begin{align}
\bar{F}_{\omega_i[k+\ell]}^{\ell}(\omega|\varphi,0)&\!=F_{\Gamma}^{\ell}\!\left(\!\frac{(1+\bar{\zeta}_i)\varsigma(\varphi,\omega)}{\bar{\zeta}_i}\!+\!\ln\left(1+\bar{\zeta}_i\right)^{\ell}\right)\label{eq:comple_CDF_0_deriv}\\
\bar{F}_{\omega_i[k+\ell]}^{\ell}(\omega|\varphi,1)&=F_{\Gamma}^{\ell}\!\left(\!\frac{\varsigma(\varphi,\omega)}{\bar{\zeta}_i}\!+\!\ln\left(1+\bar{\zeta}_i\right)^{\ell}\right)\label{eq:comple_CDF_1_deriv}
\end{align}
where $F_{\Gamma}^{\ell}(x)$ is the CDF of the sum of $\ell$ exponential random variables with unitary mean.
The key step to derive \eqref{eq:comple_CDF_0_deriv}-\eqref{eq:comple_CDF_1_deriv} is noting that, conditioned on the state $s$ of the resource, we have that $\frac{o[k]}{\theta_s^i}\overset{i.i.d.}{\sim}Exp(1)$.
Once we have a closed expression for $\bar{F}_{\omega_i[k+\ell]}^{\ell}(\omega|\varphi,s)$ for $s=0,1$ and expressions for $r_i,\rho_i$ in \eqref{eq:r_i_def}-\eqref{eq:rho_i_def} respectively, we can follow Algorithm \ref{alg:threshold_approximation} to find bounds on the thresholds and use Algorithm \ref{alg:heuristic}.
\subsection{Marketing Strategy}\label{sec:new_application}
Another potential application is in marketing. Consider a retail company trying to decide, from a set of $N$ potential products, which to buy from wholesales for a new season. The products are seasonal and lose their value at the end of the season. The length of the season is $L$ days. 
At the beginning of the season, the probability of a successful sale of product $i$ to a typical customer is $\omega_i[0]$. The company can influence the market by advertising product $i$ with a daily advertising cost of $c_i$. Assume that under daily advertising, the probability of successful sale evolves according to a Markov process with state space $[0,1]$. After $k$ ($k=1,2,\ldots$) days of advertising, assume that the probability of successful sale, $\omega_i[k]$, can be inferred, for example, from the number of clicks of the ads \cite{DBLP:journals/corr/KrishnamurthyAB16}. 
The decisions the company needs to make is when to stop advertising and if yes, whether to abandon product $i$ or buy from the whole sale and make it available for its customers (see Fig.\ref{fig:market_application_showcase}). 
\begin{figure}
\centering
\includegraphics[scale=0.25]{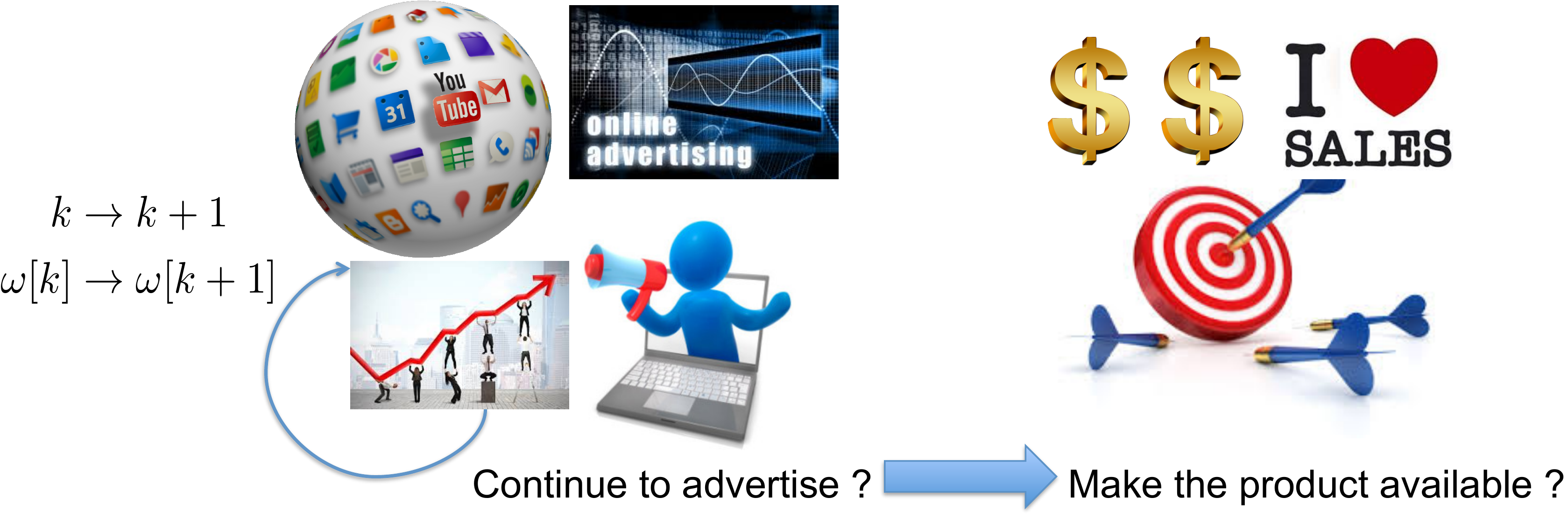}
\caption{Marketing Strategy Application}\label{fig:market_application_showcase}
\end{figure} 
Suppose that after $k$ days of advertising, the company chooses to make product $i$ available. With stochastically static customer arrivals, the expected profit from product $i$ per day for the remaining  $L-k$ days of the season is $\left[\omega_i[k] r_i -(1-\omega_i[k])\rho_i\right]$, where $r_i$ is the difference between the retail and the wholesale prices and $\rho_i$ is the daily holding cost of the unsold product. We assume that the company will stop advertising after deciding to buy in from wholesale and the probability of successful sale will stay at $\omega_i[k]$. 
Based on the well established equivalence of a POMDP with partially observable states (i.e., $s_i$ in the previous sections) to an MDP with belief (i.e., $\omega_i[k]$) as the fully observed states, we can see that the sequential decision problem studied in the previous sections apply to this marketing application assuming that the company advertises no more than one product each day and the Markovian transitions of $\omega_i[k]$ under daily advertising satisfies \eqref{eq:belief_update}.
In this case $o[k]$ is a latent random variable (not an observation) drawn from the mixture $\omega_i[k]f^i_0(o)+(1-\omega_i[k])f^i_1(o)$ (where $i$ is the advertised product at time $k$) such that 
\be\label{eq:belief_update_market}
\omega_i[k+1]=\frac{\omega[k]f^i_0(o[k])}{\omega[k]f^i_0(o[k])+(1-\omega_i[k])f^i_1(o[k])}.
\ee 
From \eqref{eq:belief_update_market} we have that the structure of the value function applies to this model as well, where for this case $s_i$ is not an underlying binary state of the product, but
is the outcome of the daily Bernoulli experiment in which the probability $\omega_i$ models the willingness of the customers to buy that product.
Under this assumption the two probability distribution functions $f^i_0(o),f^i_1(o)$ are interpreted as a parameter model fitting for the behavior of $\omega_i$ from previous market data series relative to similar products.
The motivation for assuming such model for the marketing application is that it is a relative simple model to fit but at the same time has a high flexibility given that $f_0$ and $f_1$ are general functions with no restriction other than being pdfs.
Furthermore the behavior of the function in \eqref{eq:belief_update_market} is such that $\omega$ is more sensitive to changes when it takes intermediate values (i.e. around $\frac{1}{2}$) and tends to saturate at the edges, i.e. when it approaches either $0$ or $1$.
Indeed, it is reasonable to assume the advertisement can be more impactful if the customer's uncertainty is higher, rather then when there is a clear indication of interest, in one way or the other, for a specific product.
Also, this model takes into account the avalanche effect that a popular product tends to become more and more popular: if $\omega_i[k]$ increases, it is more likely $o[k]$ will be drawn from $f^i_0(o)$ (higher weight in the mixture) and this will increase $\omega_i[k]$ and so on.
Therefore, since the action space and the expected utility correspond to the one described in our framework, our model applies. 
The bounds for the $V_t$ function need to be modified to account for the absence of a true underlying state. 
The upper bound becomes:
\begin{align}
&\overline{V_t^i}(\omega_i,i,k)=-c+\min\left\{\sum_{\ell=k}^{L-2}\prod\limits_{m=k+1}^{\ell}\hspace{-0.2cm}\overline{\mu_i}[m]\max\left\{-u[\ell-k]c\right.\right.\nonumber\\
&+(L-\ell-1)\hspace{-0.1cm}\left[\min\left\{\overline{\nu_0^i}[\ell] (r_i+\rho_i)-\rho_i \bar{F}^1_i\hspace{-0.1cm}\left(\overline{\nu_0^i}[\ell+1]|\overline{\nu_0^i}[\ell],0\right)\hspace{-0.1cm},\right.\right.\nonumber\\
&\left.\left.\left.\left. r_i\bar{F}^1_i\hspace{-0.1cm}\left(\underline{\nu_0^i}[\ell+1]|\overline{\nu_0^i}[\ell],0\right)\right\}\right],0\right\},(L-k-1)r_i\omega_i\right\}\label{eq:upper_bound_V_t_marketing}.
\end{align}
but a generalized lower bound is not tight and therefore, we have only used \eqref{eq:upper_bound_V_t_marketing} in our simulation.
The proof for obtaining \eqref{eq:upper_bound_V_t_marketing} is very similar to that given in Appendix D for the bounds in Lemma 5 and it is omitted for brevity. 
Notice that the bound for the expected sensing time in Lemma \ref{lemma_bound_expected_sensing_time} does not hold in this case, due to the different nature of the random variables $o[k]$'s.
Because of this, we cannot use the heuristic selection rule and, in the example shown in the simulation section, we have considered only the case in which the advertiser has a single product to decide on. 
The development of an analogous bound for the expected sensing time for this different model is left for future work.
\section{Simulation Results}\label{sec:simulation}
We now present our simulation results for the two applications presented: cognitive radio and marketing.
\vspace{-.3cm}
\subsection{Cognitive Radio}\label{sec:simulation_cognitive}
We highlight that, for the case of a single resource, we can compute the optimal decision thresholds $(\nu_1^i[k],\nu_0^i[k])$ obtained by MonteCarlo evaluation of the integral in \eqref{eq:V_t_def}.
In our simulations we will show the improvement in the utility we can get by using the decision threshold bounds obtained with our approach in Algorithm 1 instead of using the easy bounds in \eqref{eq:easy_lb_nu_1}-\eqref{eq:easy_ub_nu_0}, that are obtained by convexity of the function $V_t$.
For the cognitive radio case we have the following parameters: $\omega_i[0]=0.5, r_i=2,\rho_i=2,\bar{\zeta}_i=3$. 
\begin{figure}[ht]
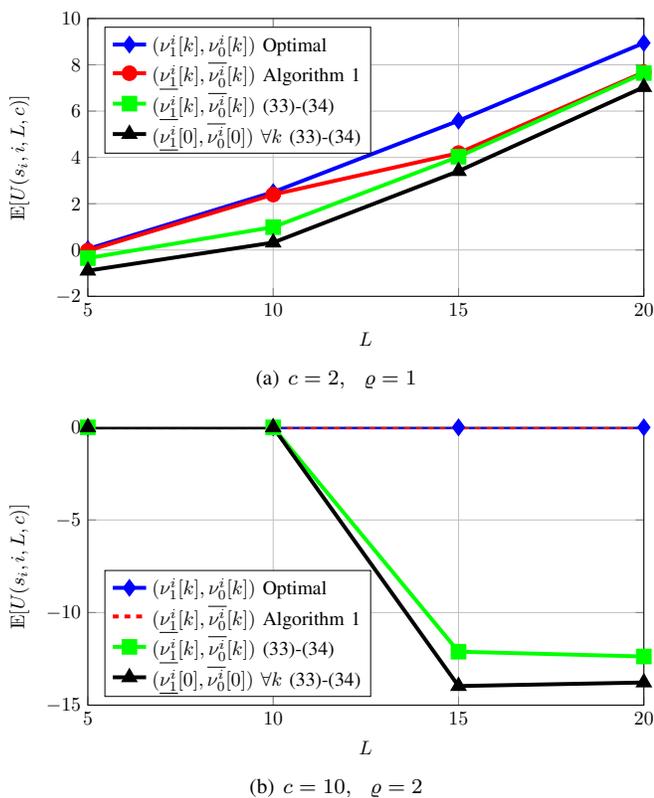

    \centering
\subfigure[$c=2,~~\varrho=1$]{
\includegraphics[width=\linewidth]{Figures/single_2_2.tikz}\label{fig:single_2_2}
}
\subfigure[$c=10,~~\varrho=2$]{
\includegraphics[width=\linewidth]{Figures/single_10_4.tikz}\label{fig:single_10_4}
}
\caption{Utility for the single-resource case with different decision thresholds in the regime of small $L$}\label{fig:utility_single_resource}
\end{figure}
In Fig. \ref{fig:single_2_2} we can see how for $L\leq 10$ our thresholds approximation in Algorithm \ref{alg:threshold_approximation} can reach the same utility of the optimal decision thresholds, while for higher values of $L$ it acquires the same utility of the easy bounds obtainable from \eqref{eq:easy_lb_nu_1}-\eqref{eq:easy_ub_nu_0}.
The reason for this is that, when $L$ increases, our upper-bound for $V_t$ in \eqref{eq:upper_bound_V_t}, which is essentially a sum of upper-bounds, becomes looser.
We also looked at the performance obtainable by constant decision thresholds (further referred to as CT strategy), i.e. 
\begin{align}
\underline{\nu_1^i}[k]&=\underline{\nu_1^i}[0]=\min\left\{\dfrac{c}{(L-1)r_i},\dfrac{\rho_i}{\rho_i+r_i}\right\}\label{eq:const_nu1}\\
\overline{\nu_0^i}[k]&=\overline{\nu_0^i}[0]=\max\left\{\dfrac{L\rho_i-c}{L\rho_i+r_i},\dfrac{\rho_i}{\rho_i+r_i}\right\}\label{eq:const_nu0}
\end{align}
$\forall k=0,1,\dots,L-1$
which always achieve worse utility than the others, being the looser approximation of the optimal thresholds.
Notice that for higher value of $c$ and $\varrho$, in Fig.\ref{fig:single_10_4}, the optimal strategy is to not start sensing and simply accepting utility $0$.
Our thresholds approximation (Algorithm \ref{alg:threshold_approximation}) is able to capture this and follows the optimal strategy, while using the bounds in \eqref{eq:easy_lb_nu_1}-\eqref{eq:easy_ub_nu_0} or constant thresholds gives a negative utility.
We can notice how the difference in utility between the optimal strategy and our approximation is higher for the marketing strategy application and developing new bounds for this case will { be the object} of future research.
For the second experiment we analyze the case with multiple channels to be sensed for the cognitive radio application. 
We consider a heterogeneous network with $4$ Primary Transmitters located at the corners of a square with side $500 m$.
Our entire bandwidth goes from $800$ to $900 MHz$ (see Fig. \ref{fig:topology_simulation}).   
A secondary transmitter wants to find opportunities to communicate with $N=20$ secondary receivers randomly spread around. 
The $20$ channels have a bandwidth of $5 MHz$ and we assume they are equally divided among the $4$ Primary transmitter: { observations are then collected at the Nyquist rate for the single channel, i.e. $2\mu s$}.
We assumed that the primary transmitter power is $10 dBm$ and the height of the transmitter is $10m$, while for the secondary we chose power equal to $5 dBm$ and a transmitter height of $3m$. 
The height of all the receivers is equal to $1m$.
\begin{figure}
\begin{center}
\includegraphics[scale=0.25]{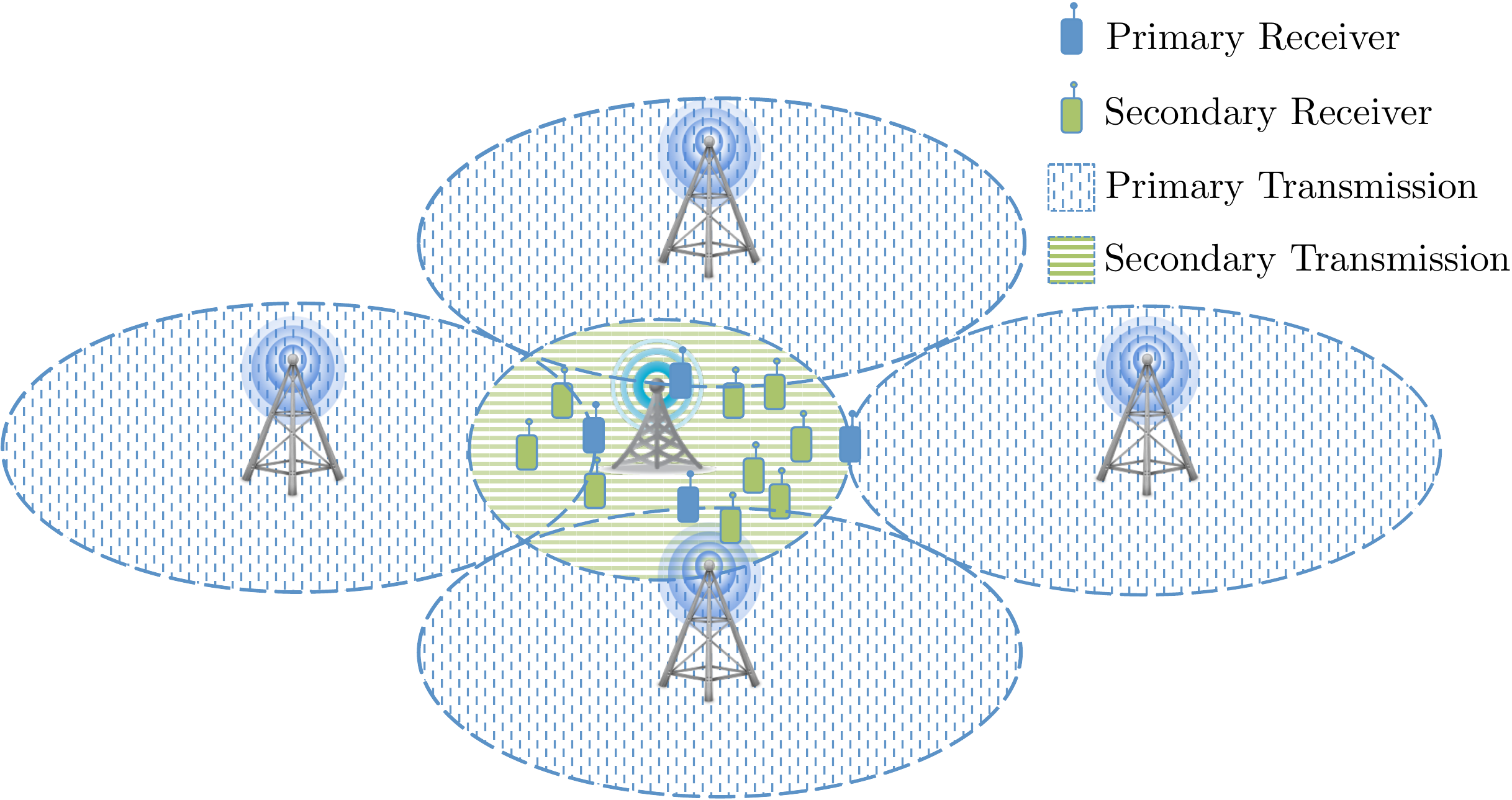}
\caption{Network Topology of our Cognitive Radio Scenario}\label{fig:topology_simulation}
\end{center}
\vspace{-0.2cm}
\end{figure}
We assume the secondary transmitter has an estimate for $\bar{\gamma}_i^S$ (the average $SNR$ at the secondary receiver from the secondary communication) and  $\bar{\zeta}_i$ (the average $SNR$ at the secondary transmitter from the primary communication), considering a deterministic path loss propagation model.
We used a deterministic two-ray model to predict the average $SNR$ received \cite{book_rappaport_wireless}, i.e. the $P_i$ value introduced in Section \ref{sec:sensing_model} $\left(\text{i.e.}~P_i=P_tG\frac{h_t^2h_r^2}{d_i^4}\right)$. The value of $G$ has been set to $2\cdot 10^{-4}$ to have average received $SNR$ ($\zeta_i$) at the secondary (in case of primary transmission) in the range of $5-10 dB$, which is a reasonable range if we assume to only have thermal noise and no additional interference over the channel.
The secondary transmitter does not know the position of the potential primary receiver but he estimates $\bar{\gamma}_i^P$ (the average $SNR$ at the primary receiver from the primary communication, without interference) and $\overline{\xi}_i$ (the average $SNR$ from the secondary transmitter to the primary receiver, that could interfere with an existing primary communication), by considering the closest point the primary receiver could be in the coverage area, i.e. the highest interference he could create (which does not necessarily correspond to the highest $\rho_i$).
As explained in Section \ref{subsec:statement}, the value of $c$ is limited by the actual cost of testing.
We want to study how the performances in terms of utility for different strategies (stopping rule and selection rule) change for different values of $\varrho$ and $c$. 
We compared the performances of Algorithm \ref{alg:heuristic} with three possible alternatives:
\begin{enumerate}
\item A selection rule that follows an initial arbitrary order (indicated in our plot with $NS$= ``No Sorting''), i.e. $\phi_k=\argmax_{i\in\A_{k+1}}i$.
\item Constant decision thresholds (indicated in our plot with $CT$= ``Constant Thresholds''), as in \eqref{eq:const_nu1}-\eqref{eq:const_nu0} $\forall i\in\N$
\item both 1) and 2) (indicated with $CTNS$)
\end{enumerate} 
Notice that the NS procedure corresponds to conducting a sequence of concatenated truncated SPRT, which, as highlighted in Section \ref{sec:formulation}, represents a suboptimal strategy for our problem (even with optimal thresholds).
We will study the performances in three different regimes of $L$.
For low values of $L$ (generally $L\leq N$) the coupling of the problem becomes more relevant since there is no time for the decision maker to sense all the resources and the single-resource decision thresholds are a loose approximation of the actual optimal thresholds which are much tighter.
In light of this, we add the following step in our Algorithm \ref{alg:heuristic}.
After computing the quantities $\overline{\mathds{E}}\left[\tau_i-k|k\right]$ for $i\in\A_k$ from \eqref{eq:E_tau_i} and sort the resources according to our index in \eqref{eq:index_sorting}, we keep in $\A_{k+1}$ all the resources with higher index as long as the following condition is satisfied:
\begin{equation}\label{eq:additional_removal}
\sum_{i\in\A_{k+1}}\overline{\mathds{E}}[\tau_i-k|k]<(1+\epsilon)(L-k-1)
\end{equation}
and add to $\D$ the remaining ones.
The motivation for this additional step is that, since we expect to not have time to sense all of them, we might start acquiring positive utility in expectation from some of them. 
In Fig.\ref{fig:small_L_AR} we indicate this additional removal with ``AR''.
\begin{figure}[ht]
\includegraphics[width=\linewidth]{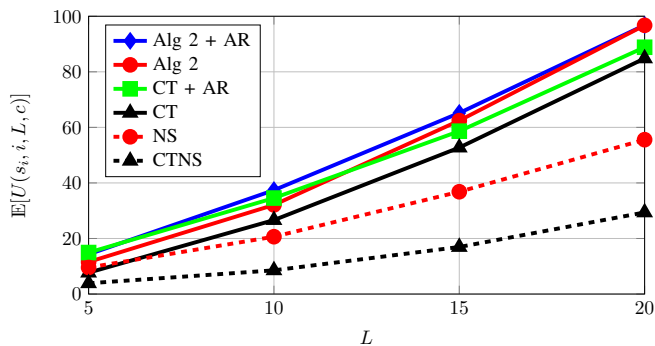}
\caption{Utility of our heuristic in the regime of small $L$ ($N=20$, $c=1,\varrho=1$ and $\epsilon=0.5$ in \eqref{eq:additional_removal}).}\label{fig:small_L_AR}
\end{figure}
We can see that for $L\leq 10$ the CT modification of Algorithm 2 can achieve higher utility, by including the additional removal of resources, than the original Algorithm \ref{alg:heuristic}. 
For $L=20$ the importance of the additional removal becomes almost negligible for Algorithm \ref{alg:heuristic}, while it continues to have an impact on the CT modification.
In this regime of L, we can see the importance of the selection rule by looking at the two strategies that do not sort the resources (dashed lines) and are clearly outperformed by the other strategies.
In Fig. \ref{fig:mod_L} we show the utility achieved by the $4$ strategies previously introduced: Alg \ref{alg:heuristic}, CT, NS, CTNS for moderate values of $L$ (i.e. $70\leq L\leq 100$).
We can see how Algorithm \ref{alg:heuristic} outperforms the other strategies and also that when we increase the cost $c$, the performances of CT, initially close to our complete heuristic, get worse and considering time-varying threshold becomes more important than sorting, i.e. the NS utility is higher than CT. 
The CTNS approach is always the worse.
The same trend was observed for different values of $\varrho$.
\begin{figure}[ht]
    \centering
\subfigure[$c=1,~~\varrho=2$]{
\includegraphics[width=\linewidth]{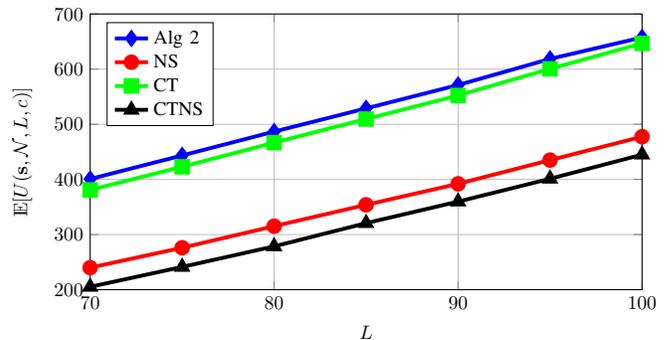}\label{fig:mod_L_1_2}
}
\subfigure[$c=5,~~\varrho=2$]{
\includegraphics[width=\linewidth]{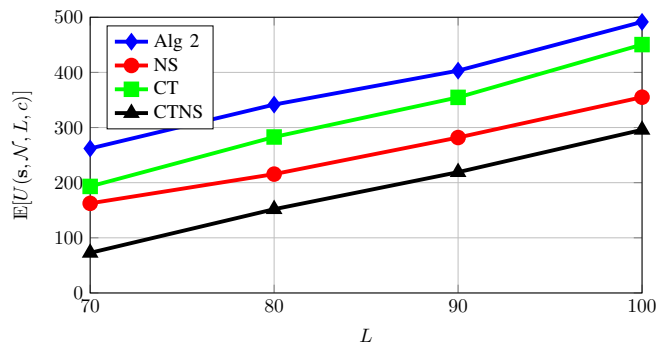}\label{fig:mod_L_5_2}
}
\subfigure[$c=10,~~\varrho=2$]{
\includegraphics[width=\linewidth]{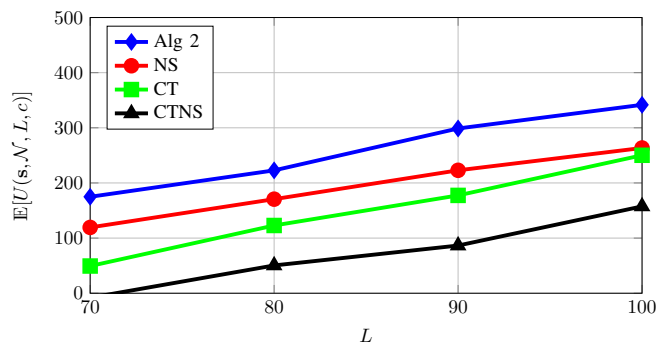}\label{fig:mod_L_10_2}
}
\caption{Utility of our heuristic in the regime of moderate $L$}\label{fig:mod_L}
\end{figure}
Finally, we look at the regime for high $L$.
Following the discussion in Section \ref{sec:asympt_regret}, in Fig.\ref{fig:Lar_L} we plot the growth rate of the regret $\frac{\Psi(L)}{\log{L}}$ for the $4$ different strategies. 
We can see this quantity is approximately constant, and the key factor to reduce this constant relies in the selection rule, other than the time-varying behavior of the decision thresholds.
Similar trends were observed for different values of $c$ and $\varrho$ where both higher $c$ and $\varrho$ increase the regret $\Psi(L)$.
\begin{figure}[ht]
\includegraphics[width=\linewidth]{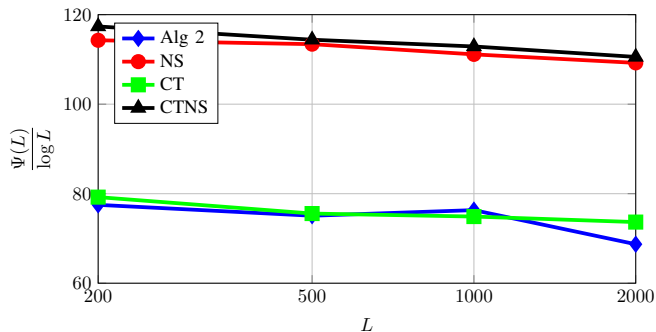}
\caption{Asymptotic growth rate of the regret $\Psi(L)$ with $\log{L}$ for our heuristic ($c=2,\varrho=2$)}\label{fig:Lar_L}
\end{figure}
\vspace{-.3cm}
\subsection{Marketing}\label{sec:simulation_market}
For the second application of online marketing strategy, as motivated in Section \ref{sec:new_application}, we can only show the case for a single product. 
We assumed $f_0(o)\sim\mathcal{N}(0;1), f_1(o)\sim\mathcal{N}(0.75;1)$ and the following parameters: $\omega_i[0]=5,r_i=1,\rho_i=1$ (where $r_i,\rho_i,c$ are normalized over their unit measure, i.e. thousands of \textdollar). 
In Fig.\ref{fig:market} we have plotted, similarly to what has been shown in Fig. \ref{fig:utility_single_resource}, the improvement in achievable utility obtained by our approximation for the bounds of $V_t$ compared to the bounds in \eqref{eq:easy_lb_nu_1}-\eqref{eq:easy_ub_nu_0}.
\begin{figure}[ht]
\includegraphics[width=\linewidth]{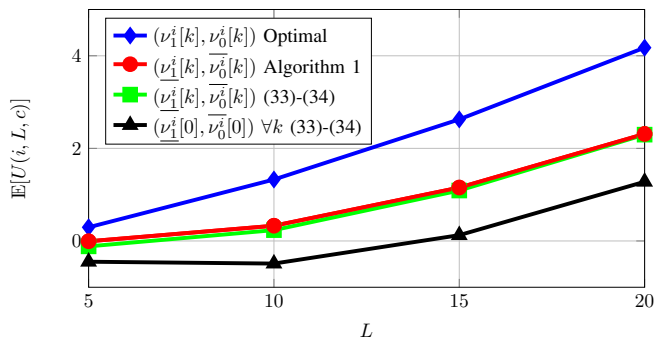}
\caption{Utility for the single-product case with different decision thresholds in the regime of small $L$ for $c=0.3$ for the marketing strategy application}\label{fig:market}
\end{figure}
We highlight that for $L=5$ only the optimal decision thresholds can guarantee a positive utility, however our thresholds approximation algorithm achieves a loss which is only $6\%$ of the loss that one would obtain using the bounds in \eqref{eq:easy_lb_nu_1}-\eqref{eq:easy_ub_nu_0} and for $L=10$ our approximation algorithm utility is $40\%$ higher.
\section{Conclusions}\label{sec:conclusions}
We proposed a new framework to design a sequence of tests over a finite time horizon and accrue a utility which is function of an unknown state vector and depends on the actions the decision maker takes at each time instant.
We derived the structure of the optimal strategy and we proposed a heuristic to approximate the optimal decisions.
Simulation results were reported to support our findings and have a physical understanding of the role played by the parameters in the utility accrued by the decision maker for two different applications: cognitive radio and marketing strategy.
\vspace{-.3cm}
\begin{appendix}
\subsection{Proof of Lemma \ref{lemma_convexity_V_i}}\label{app:proof_convexity_V_i}
We first notice that $\forall\bomega,\A$, $V(\bomega,\A,L)=0$ implies $V^j_t(\bomega,\A,L-1)=-c,~\forall~j\in\A$  and therefore 
\begin{align*}
V(\bomega,\A,L-1)&=\max_{\D\subseteq\A}\{V_d(\bomega,\D)+\max_{j\in\A\setminus\D}V^j_t(\bomega,\A,L-1)\}\\&=V_d(\bomega,\A)=\sum_{j\in\A}V^j_d(\omega_j) 
\end{align*}
which is convex in any $\omega_j, j\in\A$, since it is a positive sum of piece-wise linear functions (see \eqref{eq:V_d_def_no_delta}).
We then prove the lemma by induction.
We assume $V(\bomega,\A,k+1)$ is convex in $\omega_j$ and we show $V(\bomega,\A,k)$ is convex by showing all the possible functions $V^i_t(\bomega,\A,k)$ are convex in $\omega_j$.
We start by the case $j=i$.
Without loss of generality, let us consider two $\bomega$ say $(\bomega^1,\bomega^2)$ which differ only in the $j$-th entry (i.e $\omega_{\ell}^1=\omega_{\ell}^2, \forall \ell\neq i$). 
We want to prove that for $0\leq\lambda\leq1$ we have $\lambda V^i_t(\bomega^1,\A,k)+(1-\lambda)V^i_t(\bomega^2,\A,k)\geq V^i_t(\lambda\bomega^1+(1-\lambda)\bomega^2,\A,k)$.
From \eqref{eq:V_t_def} we can write (we will use the short notation $\bomega^{\ell,o}=\Pi(\bomega^{\ell},o,i)$ for $\ell=1,2$)
\begin{align}
&\lambda V^i_t(\bomega^1,\A,k)+(1-\lambda)V_t^i(\bomega^2,\A,k)\nonumber\\
&=-c+\int\lambda V(\bomega^{1,o},\A,k+1)f^i_{1-\omega_i^1}(o)\nonumber\\
&~~~~~~~~~~~~~+(1-\lambda)V(\bomega^{2,o},\A,k+1)f^i_{1-\omega_i^2}(o)do\nonumber\\
&=-c+\int\left[\mu V(\bomega^{1,o}_i,\A,k+1)+(1-\mu)V(\bomega^{2,o},\A,k+1)\right]\nonumber\\
&~~~~~~~~~~~~~\left[\lambda f^i_{1-\omega_i^1}(o)+(1-\lambda)f^i_{1-\omega_i^2}(o)\right]do\nonumber\\
&\overset{(a)}{\geq}-c+\int \left[V(\mu\bomega^{1,o}+(1-\mu)\bomega^{2,o},\A,k+1)\right]\nonumber\\
&~~~~~~~~~~~~~\left[\lambda f^i_{1-\omega_i^1}(o)+(1-\lambda)f^i_{1-\omega_i^2}(o)\right]do \label{eq:convexity_proof_intermediate}
\end{align}
where $(a)$ follows from the assumption that $V(\bomega,\A,k+1)$ is convex in $\omega_i$ and
$
\mu=\frac{\lambda f^i_{1-\omega_i^1}(o)}{\lambda f^i_{1-\omega_i^1}(o)+(1-\lambda)f^i_{1-\omega_i^2}(o)}.
$
Now, if we define $\bomega^3=\lambda\bomega^1+(1-\lambda)\bomega^2$ we have 
\begin{equation}\label{eq:O3_derivation}
f^i_{1-\omega_i^3}(o)=\lambda f^i_{1-\omega_i^1}(o)+(1-\lambda) f^i_{1-\omega_i^2}(o)
\end{equation}
since $f^i_{1-\omega_i}(o)$ is a linear affine function of $\omega_i$ and also
\begin{align*}
&\Pi_i(\omega_i^3,o,i)=\dfrac{\omega_i^3f_0^i(o)}{f_{1-\omega_i^3}^i(o)}=\dfrac{[\lambda\omega_i^1+(1-\lambda)\omega_i^2]f_0^{(i)}(o)}{\lambda f^i_{1-\omega_i^1}(o)+(1-\lambda) f^i_{1-\omega_i^2}(o)}\nonumber\\
&=\dfrac{\lambda f^i_{1-\omega_i^1}(o)\dfrac{\omega_i^1f_0^{(i)}(o)}{f^i_{1-\omega_i^1}(o)}+(1-\lambda)f^i_{1-\omega_i^2}(o)\dfrac{\omega_i^2f_0^{(i)}(o)}{f^i_{1-\omega_i^2}(o)}}{\lambda f^i_{1-\omega_i^1}(o)+(1-\lambda) f^i_{1-\omega_i^2}(o)}\nonumber\\
&=\mu\Pi_i(\omega_i^1,o,i)+(1-\mu)\Pi_i(\omega_i^2,o,i)
\end{align*}
which implies 
\be\label{eq:mu3_derivation} 
\bPi(\bomega^3,o,i)=\mu\bPi(\bomega^1,o,i)+(1-\mu)\bPi(\bomega^2,o,i).
\ee 
Therefore, by replacing \eqref{eq:O3_derivation},\eqref{eq:mu3_derivation} in \eqref{eq:convexity_proof_intermediate} we have 
\begin{align}
&\lambda V^i_t(\bomega^1,\A,k)+(1-\lambda)V^i_t(\bomega^2,\A,k)\nonumber\\
&\geq-c+\int V(\Pi(\bomega^3,o,i),\A,k+1)f^i_{1-\omega_i^3}(o)\nonumber\\
&=V^i_t(\bomega^3,\A,k)=V^i_t(\lambda\bomega^1+(1-\lambda)\bomega^2,\A,k) 
\end{align}
and this proves the convexity of $V^i_t(\bomega,\A,k)$ in $\omega_i$.
The convexity of $V^i_t(\bomega,\A,k)$ in $\omega_j$, $j\neq i$ can be proved by considering two points $(\bomega^1,\bomega^2)$ with the same $i$-th coordinate (i.e $\omega_{i}^1=\omega_{i}^2=\omega_i$) and following similar steps as before, where this time we have 
\be\label{eq:mu3_derivation_2} 
\bPi(\bomega^3,o,i)=\lambda\bPi(\bomega^1,o,i)+(1-\lambda)\bPi(\bomega^2,o,i).
\ee 
and we do not need to introduce $\mu$ to conclude our proof.
To show the function $V(\bomega,\A_k,k)$ is convex in $\omega_j$ we rewrite the maximization over $\D$ in \eqref{eq:value_function_exp} as follows:
\be\label{eq:rewrite_value_function_2}
\!V(\bomega,\A_k,k)\!=\!\max\left\{\max\limits_{\{j\}\subseteq\D\subseteq\A_k}\!\!J(\D),\max_{\D\subseteq\A_k\setminus\{j\}}\!J(\D)\right\}
\ee
where the function $J(\D)$ is defined in \eqref{eq:set_function}.
Let us then call $f_1$ and $f_2$ the two maximizations inside \eqref{eq:rewrite_value_function_2} in the order they appear and let us omit the arguments for brevity.
Now if we see $f_1$ and $f_2$ as functions of $\omega_j$ we have that both $f_1$ and $f_2$ are convex functions of $\omega_j$.
In fact $\omega_j$ in $f_1$ appears only as argument of $V_d^j(\omega_j)$ which is a piece-wise linear function of $\omega_j$ and therefore convex, whereas in $f_2$, $\omega_j$ is an argument of the second term of $J(\D)$, which is the maximization over the index $i$ of the functions $V_t^i(\bomega,\A_{k+1},k)$ ($i,j\in\A_{k+1}$) that are all convex in $\omega_j$. Therefore the maximum is convex and we can conclude $f_2$ is convex. 
The convexity of $V(\bomega,\A_k,k)$ follows from the fact it is the maximum of two convex functions.
\vspace{-.3cm}
\subsection{Proof of Lemma \ref{lemma_time_variant_threshold}}\label{app:time_variant_threshold}
Let us first introduce the following lemma 
\begin{lemma}\label{lemma_V_k_V_km1}
$\forall i\in\N, \omega_i\in [0,1], k=0,\dots,L-2$:
\be\label{eq:V_k_V_km1}
V_t^i(\omega_i,i,k)\geq V_t^i(\omega_i,i,k+1)+V_d^i(\omega_i) 
\ee
\end{lemma}
\begin{proof}
We will prove once again by induction. First we show that \eqref{eq:V_k_V_km1} is true for $k=L-2$ (we use the short notation $\omega_i^o=\Pi_i(\omega_i,o,i)$ and $dF_i^o=f^i_{1-\omega_i}(o)do$).
\begin{align*}
&V_t^i(\omega_i,i,L-2)+c=\hspace{-0.1cm}\int\limits_{\mathbf{O}}\hspace{-0.1cm}V\left(\omega_i^o,i,L-1\right)dF_i^o=\hspace{-0.1cm}\int\limits_{\mathbf{O}}V^i_d\left(\omega_i^o\right)dF_i^o\\
&\overset{(a)}{\geq}\!V_d^i\!\left(\int\limits_{\mathbf{O}}\omega_i^odF_i^o\right)\!\!\overset{(b)}{=}\!V_d^i(\omega_i)\!=\!V_t^i(\omega_i,i,L-1)+V_d^i(\omega_i)+c
\end{align*}
where $(a)$ holds for the convexity of the function $V_d^i$ and $(b)$ holds for the martingale property of the prior belief update
\begin{align*}
\int_{\mathbf{O}}\Pi_i(\omega_i,o,i)f_{1-\omega_i}^i(o)do=\int_{\mathbf{O}}\dfrac{\omega_if_0^i(o)}{f_{1-\omega_i}^i(o)}f_{1-\omega_i}^i(o)do=\omega_i.
\end{align*}
Then we show that if \eqref{eq:V_k_V_km1} holds for $k$, then it holds for $k-1$.
\begin{align*}
&V_t^i(\omega_i,i,k-1)+c=\int\limits_{\mathbf{O}}V\left(\omega_i^o,k\right)dF_i^o\\
&=\int\limits_{\mathbf{O}}\max\left\{(L-k)V_d^i(\omega_i^o),V_t^i(\omega_i^o,i,k)\right\}dF_i^o\\
&\overset{(a)}{\geq}\int\limits_{\mathbf{O}}\max\left\{(L-k)V_d^i(\omega_i^o),V_t^i(\omega_i^o,i,k+1)+V_d^i(\omega_i^o)\right\}dF_i^o\\
&=\hspace{-0.1cm}\int\limits_{\mathbf{O}}\hspace{-0.1cm}\left[\max\left\{(L-k-1)V_d^i(\omega_i^o),V_t^i(\omega_i^o,i,k+1)\right\}+V_d^i(\omega_i^o)\right]dF_i^o\\
&=V_t^i(\omega_i,i,k)+c+\int\limits_{\mathbf{O}}V_d^i(\omega_i^o)dF_i^o\\
&\overset{(b)}{\geq}V_t^i(\omega_i,i,k)+c+V_d^i\left(\int_{\mathbf{O}}\omega_i^odF_i^o\right)\overset{(c)}{=}V_t^i(\omega_i,i,k)+c+V_d^i(\omega_i)
\end{align*}
where $(a)$ holds for the induction hypotheses, $(b)$ for the convexity of $V_d^i$ and $(c)$ for the martingale property of the belief update.
\end{proof}
Now, to show the two thresholds respect \eqref{eq:nu_1_k}-\eqref{eq:nu_0_k} it is equivalent to prove the following statement 
\begin{align*}
\forall \omega_i\in[0,1],~&(L-k)V_d^i(\omega_i)\geq V_t^i(\omega_i,i,k)\\&\Rightarrow (L-k-1)V_d^i(\omega_i)\geq V_t^i(\omega_i,i,k+1)
\end{align*}
and this can be proved since 
\begin{align*}
&(L-k-1)V_d^i(\omega_i)=(L-k)V_d^i(\omega_i)-V_d^i(\omega_i)\\
&\overset{(a)}{\geq} V_t^i(\omega_i,i,k)-V_d^i(\omega_i)\overset{(b)}{\geq}V_t^i(\omega_i,i,k+1)
\end{align*}
where $(a)$ holds by hypothesis and $(b)$ follows from Lemma \ref{lemma_V_k_V_km1}, and this completes the proof.
\subsection{Proof of Lemma \ref{lemma_threshold_inequalities}}\label{app:threshold_inequalities}
We start our proof by highlighting a property of the value function defined in \eqref{eq:value_function_exp} with the following lemma
\begin{lemma}\label{lemma_submodular_value_function}
$\forall i\in\N,\A'\in 2^{\N-i},\omega_i\!\in\![0,1],k=0,\dots,L-1$
\be\label{eq:submodular_value_function}
V(\bomega,\A',k)+V(\omega_i,i,k)\geq V(\bomega,\A'+i,k) 
\ee
\end{lemma}
\begin{proof}
Let us consider the function $J(\D)$ defined in \eqref{eq:set_function} refers to the value function when $\A_k=\A=\A'+i$ and $J'(\D)$ refers to the value function when $\A_k=\A'$.
We will again use induction. 
Clearly \eqref{eq:submodular_value_function} holds for $k=L-1$ where 
\begin{equation*}
\left[V_d(\bomega,\A')+V_d^i(\omega_i)\right]\geq V_d(\bomega,\A)
\end{equation*} 
Then, assuming \eqref{eq:submodular_value_function} holds for $k+1$, we prove it holds for $k$. 
Let us write  \eqref{eq:submodular_value_function} as
\be
\max_{\D\subseteq\A'}J'(\D)+V(\omega_i,i,k)\geq \max_{\D\subseteq\A}J(\D) 
\ee
and call $\tilde{\D}=\argmax\limits_{\D\subseteq\A}J(\D)$.
There are three possible cases:
\begin{enumerate}
\item If $i\in\tilde{\D}$,
\begin{align*}
&\max_{\D\subseteq\A'}J'(\D)+V(\omega_i,i,k)\geq J'(\tilde{\D}-i)+(L-k)V_d^i(\omega_i)\\&=J(\tilde{\D})=\max_{\D\subseteq\A}J(\D) 
\end{align*}
\item If $i\notin\tilde{\D}$ and $i=\argmax\limits_{j\in(\A)\setminus\tilde{\D}}V_t^j(\bomega,\A\setminus\tilde{\D},k)$
\begin{align*}
&\max_{\D\subseteq\A'}J'(\D)+V(\omega_i,i,k)\\
&\overset{(a)}{\geq} (L-k)V_d(\bomega,\tilde{\D})+V(\bomega,\A'\setminus\tilde{\D},k)+V_t^i(\omega_i,i,k)\\
&\geq (L-k)V_d(\bomega,\tilde{\D})+V(\bomega,\A'\setminus\tilde{\D},k+1)-c\\&+\mathop{{\int}}\limits_{\mathbf{O}}V\left(\omega_i^o,i,k+1\right)dF_i^o\overset{(b)}{\geq}(L-k)V_d(\bomega,\tilde{\D})-c\\&+\mathop{\mathlarger{\int}}\limits_{\mathbf{O}}V\left(\bPi(\bomega,o,i),\A\setminus\tilde{\D},k+1\right)dF_i^o\\
&=(L-k)V_d(\bomega,\tilde{\D})+V_t^i(\bomega,\A\setminus\tilde{\D},k)=\max_{\D\subseteq\A}J(\D) 
\end{align*}
where $(a)$ follows from the definition of our value function in \eqref{eq:value_function_exp} and $(b)$ holds by the induction hypothesis.
\item If $i\notin\tilde{\D}$ and $i\neq j^*=\argmax\limits_{j\in\A\setminus\tilde{\D}}V_t^j(\bomega,\A\setminus\tilde{\D},k)$
the proof is similar to the previous case.
\end{enumerate}
\end{proof}
Proving \eqref{eq:threshold_inequality_1}-\eqref{eq:threshold_inequality_0} is equivalent to prove:
\begin{align*}
\forall \omega_i\in[0,1],~&(L-k)V_d^i(\omega_i)\geq V_t^i(\omega_i,i,k)\\
\Rightarrow &\max\limits_{\{i\}\subseteq\D\subseteq\A}J(\D)\geq\max_{\D\subseteq\A'}J(\D).
\end{align*}
To prove this, we call $\D^*=\argmax\limits_{\D\subseteq\A'}J'(\D)$ and use the chain of inequalities
\begin{align*}
&\max\limits_{\{i\}\subseteq\D\subseteq\A}J(\D)\geq J(\D^*+i)=(L-k)V_d^i(\omega_i)\\&~~+(L-k)V_d(\bomega,\D^*)+\hspace{-0.2cm}\max\limits_{j\in\A'\setminus\D^*}\hspace{-0.2cm}V_t^j(\bomega,\A'\setminus\D^*,k)\\
&\overset{(a)}{=} V(\omega_i,i,k)+V(\bomega,\A',k)\overset{(b)}\geq V(\bomega,\A,k)\overset{(c)}{\geq}\max_{\D\subseteq\A'}J(\D)
\end{align*}
where $(a)$ holds by hypotheses, since if $(L-k)V_d^i(\omega_i)\geq V_t^i(\omega_i,i,k)$ then $V(\omega_i,i,k)=(L-k)V_d^i(\omega_i)$ and definition of $J'(\D)$, $(b)$ follows by Lemma \ref{lemma_submodular_value_function} and $(c)$ by definition of value function in \eqref{eq:value_function_exp}, and this completes the proof.
\vspace{-.3cm}
\subsection{Proof of Lemma \ref{lemma_threshold_approximation}}\label{app:threshold_approximation}
Since we are focusing on a specific resource $i$ next, for brevity we drop the index $i$ and define the following events:
\begin{align*}
&A_k\triangleq\{\omega[k]\geq\nu_0[k]\},~B_k\triangleq\{\omega[k]\leq\nu_1[k]\},\\&C_k\triangleq \overline{A_k\cup B_k},~\C_{k}^{\ell}=\bigcap\limits_{m=k}^{\ell}C_m
\end{align*}
We then write the function $V_t(\omega,k)$ as follows (we use the short notation $P_0(A)=P(A|s=0)$ and same for $P_1(A)$):
\begin{align}
&V_t(\omega,k)=-c+\sum_{\ell=k}^{L-2}P\left(\C_{k+1}^{\ell}\right)\mathds{E}\left[-u[\ell-k]c+(L-\ell-1)\right.\nonumber\\&\left.\left[\omega[\ell]rP_0\left(A_{\ell+1}\right)\hspace{-0.05cm}-\hspace{-0.05cm}(1-\omega[\ell])\rho P_1\left(A_{\ell+1}\right)\right]\big|\C_{k+1}^{\ell}\right]\label{eq:V_t_prob}
\end{align}
To find bounds for $V_t(\omega,k)$ we will use:
\begin{enumerate}
\item an upper/lower bound for $P\left(\C_{k+1}^{\ell}\right)$
\item an upper/lower bound for $P_0\left(A_{\ell+1}\big|\C_{k+1}^{\ell}\right)$
\item an upper/lower bound for  $P_1\left(A_{\ell+1}\big|\C_{k+1}^{\ell}\right)$
\end{enumerate} 
First we write:
\begin{align}
&P\left(\C_{k+1}^{\ell}\right)=P\left(C_{\ell} \big|\C_{k+1}^{\ell-1}\right)P\left(\C_{k+1}^{\ell-1}\right)\nonumber\\
&=\left[1-P\left(A_{\ell}\big| \C_{k+1}^{\ell-1}\right)-P\left(B_{\ell} \big| \C_{k+1}^{\ell-1}\right)\right]P\left(\C_{k+1}^{\ell-1}\right)\label{eq:recursion}
\end{align}
We now find 2). 3) follows a similar approach and 1) will be otained using 2) and 3) (in our derivation with $P(\cdot|\varphi)$ we indicate we are conditioning on $\omega[\ell]=\varphi$).
\begin{align*}
&P_0\left(A_{\ell+1}\big| \C_{k+1}^{\ell}\right)=\int\limits_{C_{\ell}} P_0\left(A_{\ell+1}|t\right)\dfrac{P_0\left(t\big|\C_{k+1}^{\ell-1}\right)}{P_0\left(C_{\ell}\big|\C_{k+1}^{\ell-1}\right)}dt\\
&\overset{u.b.}{\leq}\hspace{-0.2cm}\max_{\nu_1[\ell]<t<\nu_0[\ell]}\hspace{-0.3cm}P_0\left(A_{\ell+1}|t\right)=P_0\left(A_{\ell+1}\big|\nu_0[\ell]\right)\leq P_0\left(A_{\ell+1}\big|\overline{\nu_{0}}[\ell]\right)\\&~~~=\bar{F}_{\omega[\ell+1]}^1\left({\nu_{0}}[\ell+1]|\overline{\nu_{0}}[\ell],0\right)\leq \bar{F}_{\omega[\ell+1]}^1\left(\underline{\nu_{0}}[\ell+1]|\overline{\nu_{0}}[\ell],0\right)\\
&\overset{l.b.}{\geq}\hspace{-0.2cm}\min_{\nu_1[\ell]<t<\nu_0[\ell]}\hspace{-0.3cm}P_0\left(A_{\ell+1}|t\right)=P_0\left(A_{\ell+1}|\nu_1[\ell]\right)\geq P_0\left(A_{\ell+1}|\underline{\nu_{1}}[\ell]\right)\\
&~~~~=\bar{F}^1_{\omega[\ell+1]}\left({\nu_{0}}[\ell+1]|\underline{\nu_{1}}[\ell],0\right)\geq\bar{F}^1_{\omega[\ell+1]}\left(\overline{\nu_{0}}[\ell+1]|\underline{\nu_{1}}[\ell],0\right)
\end{align*}
where the upper and lower bound start by taking out of the integral the maximum and minimum value of \\ $P_0(A_{\ell+1}|\omega[\ell]=t)$ and notice that 
\be
\int\limits_{C_{\ell}}P_0\left(\omega[\ell]=t\big|\C_{k+1}^{\ell-1}\right)dt=P_0\left(C_{\ell}\big|\C_{k+1}^{\ell-1}\right)
\ee  
Following similar steps for the event $B_{\ell+1}$ we can derive:
\begin{align}
P_0\left(A_{\ell+1}\big|\C_{k+1}^{\ell}\right)
&\overset{u.b.}{\leq}\bar{F}_{\omega[\ell+1]}^1\left(\underline{\nu_{0}}[\ell+1]|\overline{\nu_{0}}[\ell],0\right)\label{eq:ub_A_ell}\\
P_0\left(A_{\ell+1}\big|\C_{k+1}^{\ell}\right)&\overset{l.b.}{\geq}\bar{F}^1_{\omega[\ell+1]}\left(\overline{\nu_{0}}[\ell+1]|\underline{\nu_{1}}[\ell],0\right)\label{eq:lb_A_ell}\\
P_0\left(B_{\ell+1}\big|\C_{k+1}^{\ell}\right)&\overset{u.b.}{\leq}1-\bar{F}_{\omega[\ell+1]}^1\left(\overline{\nu_1}[\ell+1]|\underline{\nu_1}[\ell],0\right)\label{eq:ub_B_ell}\\
P_0\left(B_{\ell+1}\big|\C_{k+1}^{\ell}\right)&\overset{l.b.}{\geq}1-\bar{F}_{\omega[\ell+1]}^1\left(\underline{\nu_1}[\ell+1]|\overline{\nu_0}[\ell],0\right)\label{eq:lb_B_ell}
\end{align}
Exactly the same bounds can be found for the case $s=1$, replacing the index of $P$ and the conditioned state on $F$ from $0$ to $1$.
For the case of $\ell=k$, which corresponds to the probabilities in the first term of the sum in \eqref{eq:V_t_prob}, we can use the bounds in \eqref{eq:ub_A_ell}-\eqref{eq:lb_A_ell} by considering 
$\overline{\nu_{0}}[k]=\underline{\nu_1}[k]=\omega$.
We can then use the recursion in \eqref{eq:recursion} and the bounds previously found to derive:
\begin{align}
P\left(\C_{k+1}^{\ell}\right)&\overset{u.b.}{\leq}\omega\hspace{-0.2cm}\prod\limits_{m=k+1}^{\ell}\overline{\mu}[m|0]+(1-\omega)\hspace{-0.2cm}\prod\limits_{m=k+1}^{\ell}\overline{\mu}[m|1]\\
P\left(\C_{k+1}^{\ell}\right)&\overset{l.b.}{\geq}\omega\hspace{-0.2cm}\prod\limits_{m=k+1}^{\ell}\underline{\mu}[m|0]+(1-\omega)\hspace{-0.2cm}\prod\limits_{m=k+1}^{\ell}\underline{\mu}[m|1]
\end{align}
Once we upper and lower bound the probabilities for each term of the sum in \eqref{eq:V_t_prob}, then the argument of the expectation is a monotonically increasing function of $\omega[\ell]$. Therefore, since we are conditioning on the event $\C_{k+1}^{\ell}$, the expectation will be lower or upper-bounded by choosing $\omega[\ell]=\underline{\nu_1}[\ell]$ or $\omega[\ell]=\overline{\nu_0}[\ell]$ respectively. 
Gathering these results and considering the convexity of $V_t^i(\omega_i,i,k)$, we can write the upper and lower bound $\overline{V_t}(\omega,k),\underline{V_t}(\omega,k)$ as in \eqref{eq:upper_bound_V_t}-\eqref{eq:lower_bound_V_t} where we introduce the $\max$ to make sure our functions remain above $c$.
\vspace{-.3cm}
\subsection{Proof of Lemma \ref{lemma_bound_expected_sensing_time}}\label{app:proof_bound_expected_sensing_time}
Without loss of generality we can prove the lemma at $k=0$ and for general $k$ simply consider the time translation and the equivalent problem for $k'=0,1,\dots,L-k-1$.
For the remainder of the proof the index $k$ is used as the time index for the test of resource $i$ that starts at time $0$.
We indicate with $\Lambda^{k}_i\triangleq\sum_{\ell=1}^{k}\log\left(\frac{f_1^i(o[\ell-1])}{f_0^i(o[\ell-1])}\right)$ the log-likelihood of the samples collected for resource $i$ up to time $k$.
By considering the belief update in \eqref{eq:belief_update} and the optimal policy structure in Theorem \ref{th:optimal_policy_structure}, a final decision on resource $i$ can be made as soon as
\begin{equation}\label{eq:stop_likelihood}
\Lambda_i^{k} \geq \Upsilon_1^i\left(\A_k,k\right)~ \vee~ \Lambda_i^{k} \leq \Upsilon_0^i\left(\A_k,k\right)
\end{equation}
where $\Upsilon^i_{0,1}(\A_k,k)\triangleq\varsigma(\omega_i[0],\nu_{0,1}^i(\A_k,k))$, and we recall the resources in the set $\A_k$ are the resources for which a decision is still pending at time $k$. 
From \eqref{eq:stop_likelihood} it follows that $\tau_i$ 
\begin{align} 
&\tau_i\triangleq \inf\left\{0\leq k\leq L-1\hspace{-0.05cm}:\Lambda_i^{k}\hspace{-0.05cm}\geq\hspace{-0.05cm}\Upsilon_1^i\left(\A_k,k\right)\hspace{-0.05cm}\vee\hspace{-0.05cm}\Lambda_i^{k}\hspace{-0.05cm}\leq\hspace{-0.05cm}\Upsilon_0^i\left(\A_k,k\right)\right\}
\end{align} 
is a stopping time.
From \eqref{eq:threshold_inequality_nu_1}-\eqref{eq:threshold_inequality_nu_0} with $k'=0$ we can deduce $\forall k=0,1,\dots,L-1,\forall \A_k\subseteq\N$
\begin{align}
\Upsilon^i_{0}(\A_k,k)&\geq\varsigma(\omega_i[0],\overline{\nu_0^i}[0])=\underline{\Upsilon^i_0}\label{eq:Upsilon_bound_0}\\
\Upsilon^i_{1}(\A_k,k)&\leq\varsigma(\omega_i[0],\underline{\nu_1^i[0]})=\overline{\Upsilon^i_1}\label{eq:Upsilon_bound_1}
\end{align}
Therefore 
\begin{align} 
&\overline{\tau}_i\triangleq \inf\left\{0\leq k\leq L-1:\Lambda_i^{k}\geq\overline{\Upsilon_1^i}\vee\Lambda_i^{k}\leq\underline{\Upsilon_0^i}\right\}
\end{align} 
is a stopping time always greater than $\tau_i$, i.e. $P(\overline{\tau}_i>\tau_i)=1$.
We then derive a bound on the expected value of $\overline{\tau}_i$ and this will also hold  for $\tau_i$.
We will use a similar technique as the one used for the thresholds approximation (see Proof of Lemma \ref{lemma_threshold_approximation} in Appendix \ref{app:threshold_approximation}) to write
\begin{align}
\mathds{E}_1\left[{\Lambda_i^{\overline{\tau}_i}}\right]&=E_1\left[\Lambda_i^{\overline{\tau}_i-1}+\Lambda_i^1\right]\leq \overline{\Upsilon_1^i}+\mathds{E}_1[\Lambda_i^1|\Lambda_i^1>0]\\
\mathds{E}_0\left[{\Lambda_i^{\overline{\tau}_i}}\right]&=E_0\left[\Lambda_i^{\overline{\tau}_i-1}+\Lambda_i^1\right]\geq \underline{\Upsilon_0^i}+\mathds{E}_0[\Lambda_i^1|\Lambda_i^1<0]
\end{align}
where the bounds follow by $\overline{\tau_i}$ being the stopping time, therefore we know $\underline{\Upsilon_0^i}<\Lambda_i^{\overline{\tau_i}-1}<\overline{\Upsilon_1^i}$ and $\Lambda_i^{\overline{\tau_i}}>\overline{\Upsilon_1^i}$ or $\Lambda_i^{\overline{\tau_i}}<\underline{\Upsilon_0^i}$.
By Wald's identity we have 
\be 
\mathds{E}\left[\Lambda_i^{\overline{\tau_i}}\right]=\mathds{E}[\overline{\tau_i}]\mathds{E}\left[\Lambda_i^1\right]
\ee 
and since $\mathds{E}_0\left[\Lambda_i^1\right]=-D(f_0^i||f_1^i)$, $\mathds{E}_1\left[\Lambda_i^1\right]=D(f_1^i||f_0^i)$
we can write 
\begin{align}
\mathds{E}_0\left[\overline{\tau_i}\right]&\leq\dfrac{-\underline{\Upsilon_0^i}-\mathds{E}_0[\Lambda_i^1|\Lambda_i^1<0]}{D(f_0^i||f_1^i)}\label{eq:bound_E_0}\\
\mathds{E}_1\left[\overline{\tau_i}\right]&\leq\dfrac{\overline{\Upsilon_1^i}+\mathds{E}_1[\Lambda_i^1|\Lambda_i^1>0]}{D(f_1^i||f_0^i)}\label{eq:bound_E_1}
\end{align}
and combining \eqref{eq:bound_E_0}-\eqref{eq:bound_E_1} (conditioned on the status of the resource) and evaluating the two bounds in \eqref{eq:Upsilon_bound_0}-\eqref{eq:Upsilon_bound_1} at any time $k\neq 0$ we can obtain the bound in \eqref{eq:E_tau_i}.
\vspace{-0.3cm} 
\subsection{Motivation for the index in our decision algorithm}
\label{app:index_motivation}
The reason to approximate the optimal selection rule with the index $\frac{\omega_i[k]r_i}{\mathds{E}[\tau_i-k|k]}$ finds its motivation in the asymptotic utility growth. 
Let us consider $L$ large enough such that we can neglect the probability of taking a wrong decision over a certain resource, and we can assume the sensing time for each resource is not affected by having spent time sensing other resources before. 
We also further limit our strategy to sequentially sense each resource until a decision is made and then switch to a different resource. 
Then there is an optimal sorting for this strategy which maximizes the expected utility. 
To find such sorting, we use an interchange argument. Let us consider a pair of arbitrary resources, say $1$ and $2$, and show the sorting $1,2$ at time $0$ with $L$ instants available is optimal if (we use the short notation $\mathds{E}[\tau_i]$ for $\mathds{E}[\tau_i-k|k]$ when $k=0$): 
\begin{align*}
&-c\mathds{E}[\tau_1]+(L-\mathds{E}_0[\tau_1])\omega_1r_1-c\mathds{E}[\tau_2-\tau_1|\tau_1]\\&+(L-\mathds{E}[\tau_1]-\mathds{E}_0[\tau_2-\tau_1|\tau_1])\omega_2r_2>\\&-c\mathds{E}[\tau_2]+(L-\mathds{E}_0[\tau_2])\omega_2r_2-c\mathds{E}[\tau_1-\tau_2|\tau_2]\\&+(L-\mathds{E}[\tau_2]-\mathds{E}_0[\tau_1-\tau_2|\tau_2])\omega_1r_1~~ \overset{(a)}{\Rightarrow}\\
&-c\mathds{E}[\tau_1]+(L-\mathds{E}_0[\tau_1])\omega_1r_1-c\mathds{E}[\tau_2]\\&+(L-\mathds{E}[\tau_1]-\mathds{E}_0[\tau_2])\omega_2r_2>-c\mathds{E}[\tau_2]+(L-\mathds{E}_0[\tau_2])\omega_2r_2\\&-c\mathds{E}[\tau_1]+(L-\mathds{E}[\tau_2]-\mathds{E}_0[\tau_1])\omega_1r_1~~\Leftrightarrow\\
&-\mathds{E}[\tau_1]\omega_2r_2>-\mathds{E}[\tau_2]\omega_1r_1~~\Leftrightarrow~~\dfrac{\omega_1r_1}{\mathds{E}[\tau_1]}>\dfrac{\omega_2r_2}{\mathds{E}[\tau_2]}
\end{align*}
where $(a)$ follows from the assumption the sensing time for the second resource is not affected by the time spent in sensing the previous one, which is a reasonable assumption for $L$ large enough.
However, due to the time variant threshold and the finite horizon scenario, the expected sensing time strongly depends on the time $k$ the test starts, and is therefore affected by the sorting. 
Nevertheless, this index represents a good approximation that takes into account the identifiability of the resource and simulation results will corroborate our intuition.
\vspace{-.8cm}
\subsection{Derivation of \eqref{eq:R_i_xi_P}}\label{app:derivation_R_i_xi_P}
The instantaneous $SNR$ at the primary receiver in presence of interference from the secondary transmitter can be expressed as
$\gamma_{i,\bar{\xi}_i}^P=\frac{\bar{\gamma}_i^P\left|h^P\right|^2}{1+\bar{\xi}_i\left|h^S\right|^2}$, 
where $h^P,h^S$ are the two complex channel gains and for Rayleigh fading their absolute value squared is exponential with unitary mean.
By following the same derivation as in \cite{outage_capacity_def} for the no interferer case, and assuming independence between $\left|h^P\right|^2$ and $\left|h^S\right|^2$ one can use the cdf of $Y=X_1-X_2$, where $X_1\sim Exp(\alpha_1)$, $X_2\sim Exp(\alpha_2)$: 
\be 
P\left(Y\leq y\right)=1-\dfrac{\alpha_1}{\alpha_1+\alpha_2}e^{-\frac{y}{\alpha_1}}~~~~~~~~y>0
\ee 
to derive \eqref{eq:deriv_P_i_bar_xi_i}, and consequently the effective rate for the primary receiver expressed by \eqref{eq:R_i_xi_P}.

\end{appendix}
\bibliographystyle{IEEEtran}

\vspace{-1.25cm}
\begin{IEEEbiography}
[{\includegraphics[width=1in,height=1.25in,clip,keepaspectratio]{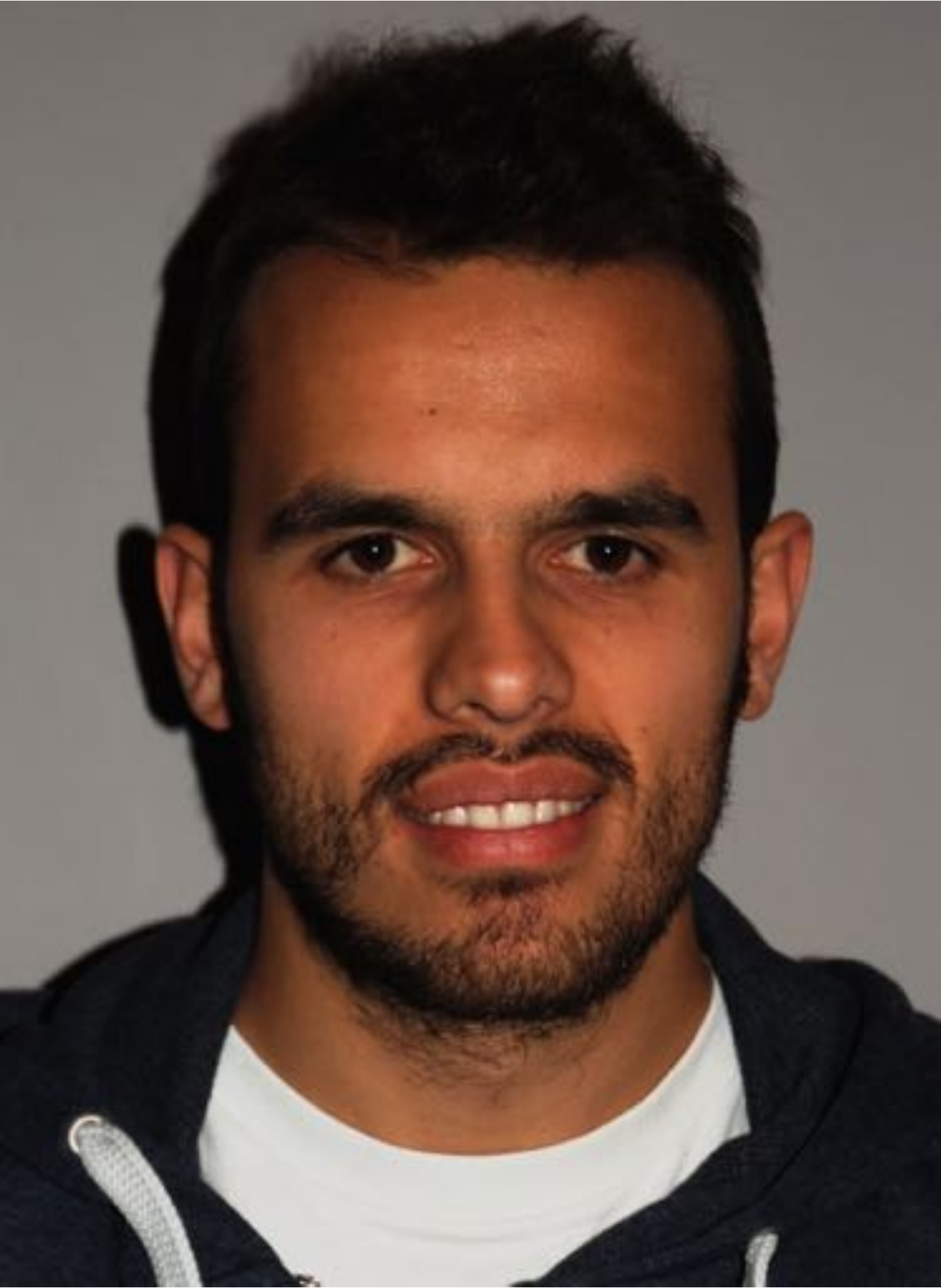}}]
{Lorenzo Ferrari} (S '14) is currently a PhD student in Electrical Engineering at Arizona State University.
Prior to that, he received his B.Sc. and M.Sc degree in Electrical Engineering from University of Modena, Italy in 2012 and 2014 respectively. His research interests lie in the broad area of wireless communications and signal processing. He has received the IEEE SmartGridComm 2014 Best Student Paper Award for the paper ``The Pulse Coupled Phasor Measurement Unit''.
\end{IEEEbiography}
\vspace{-1.25cm}
\begin{IEEEbiography}
[{\includegraphics[width=1in,height=1.25in,clip]{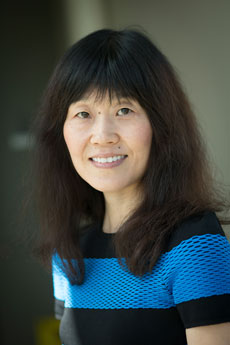}}]
{Qing  Zhao} (F '13) joined  the  School  of  Electrical  and  Computer  Engineering  at Cornell  University  in 2015 as  a  Professor.  Prior to that, she was  a  Professor  at  University  of  California,  Davis. She  received the Ph.D. degree in Electrical Engineering in 2001 from Cornell University. Her research interests are in
the general  area of stochastic optimization,  decision theory,  machine  learning, and algorithmic theory in
dynamic systems and communication and social-economic networks. She received the 2010 IEEE Signal
Processing  Magazine  Best  Paper  Award  and  the  2000  Young  Author  Best  Paper  Award  from  the  IEEE
Signal Processing Society.
\end{IEEEbiography}
\vspace{-1.25cm}
\begin{IEEEbiography}
[{\includegraphics[width=1in,height=1.25in,clip,keepaspectratio]{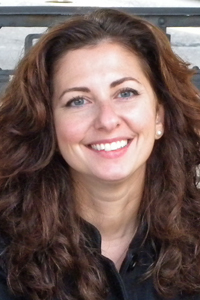}}]
{Anna Scaglione} (F '11) is currently a professor in
electrical and computer engineering at Arizona State
University. She was previously at the University of
California at Davis, Cornell University and University
of New Mexico. Her research focuses on various
applications of signal processing in network and data science
that include intelligent infrastructure for energy delivery and information
systems.
Dr. Scaglione was elected an IEEE fellow in
2011. She received the 2000 IEEE Signal Processing
Transactions Best Paper Award and more recently
was honored for the 2013, IEEE Donald G. Fink Prize Paper Award for the
best review paper in that year in the IEEE publications, her work with her
student earned 2013 IEEE Signal Processing Society Young Author Best Paper
Award (Lin Li). She was EIC of the IEEE Signal Processing Letters and
served in many other capacities the IEEE Signal Processing, IEEE
Communication societies and is currently Associate Editor for the 
IEEE Transactions on Control over Networked Systems.
\end{IEEEbiography}




\end{document}